\RequirePackage{fix-cm}
\documentclass[smallextended]{svjour3}       
\smartqed  
\usepackage{graphicx}
\usepackage{fullpage}
\usepackage{amsmath}
\usepackage{amsfonts}
\usepackage{amssymb}
\usepackage{enumerate}
\usepackage{threeparttable}
\usepackage{multirow}
\usepackage[blocks]{authblk}
\usepackage{rotating,epsfig}
\usepackage{lscape}
\usepackage{textcomp}
\usepackage{algorithm}
\usepackage{algpseudocode}
\usepackage{natbib}

\usepackage{array}
\usepackage{mathrsfs}
\usepackage{kotex}
\usepackage{xcolor}
\usepackage{hyperref}
\definecolor{OliveGreen}{rgb}{0,0.6,0}
\hypersetup{colorlinks=true, citecolor=OliveGreen}
\usepackage{ragged2e}
\usepackage{tabularx}
\definecolor{navyblue}{rgb}{0,0.352,0.674}

\newcommand{\BE}{\begin{equation}}
\newcommand{\EE}{\end{equation}}
\newcommand{\BEN}{\begin{equation}\nonumber}
\newcommand{\BA}{\begin{array}}
\newcommand{\EA}{\end{array}}
\newcommand{\BPA}{\begin{pmatrix}\begin{array}}
\newcommand{\EAP}{\end{array}\end{pmatrix}}
\newcommand{\BI}{\begin{itemize}}
\newcommand{\EI}{\end{itemize}}

\newcommand{\blue}{\textcolor{blue}}

\newcommand{\MC}{\mathcal}

\newcommand{\BEQ}{\begin{equation}}
\newcommand{\EEQ}{\end{equation}}
\newcommand{\BEA}{\begin{eqnarray}}
\newcommand{\EEA}{\end{eqnarray}}
\newcommand{\BIT}{\begin{itemize}}
\newcommand{\EIT}{\end{itemize}}

\begin{document}

\title{An efficient parallel block coordinate descent
algorithm for large-scale precision matrix estimation using graphics processing units
}


\author{Young-Geun Choi        \and
        Seunghwan Lee \and Donghyeon Yu 
}


\institute{Young-Geun Choi \at 
              Department of Statistics, Sookmyung Women's University, Seoul, Korea.
              \and
              Seunghwan Lee \at
              Department of Statistics, Inha University, Incheon, Korea. 
              \and
              Donghyeon Yu \at
              Department of Statistics, Inha University, Incheon, Korea. \\
              \email{dyu@inha.ac.kr}           
}

\date{Received: date / Accepted: date}

\maketitle

\begin{abstract}
Large-scale sparse precision matrix estimation has attracted wide interest from the statistics community.
The convex partial correlation selection method (CONCORD) developed by Khare et al. (2015)
has recently been credited with some theoretical properties for estimating sparse precision matrices.
The CONCORD obtains its solution by a coordinate descent algorithm (CONCORD-CD) based on  the convexity of the objective function. 
However, since a coordinate-wise update in CONCORD-CD is \emph{inherently serial}, a scale-up is nontrivial.
In this paper, we propose a novel parallelization of CONCORD-CD, namely, CONCORD-PCD. 
CONCORD-PCD partitions the off-diagonal elements into several groups and updates each group simultaneously without harming the computational convergence of CONCORD-CD.
We guarantee this by employing the notion of edge coloring in graph theory. Specifically, we establish a nontrivial correspondence between \emph{scheduling the updates} of the off-diagonal elements in CONCORD-CD and \emph{coloring the edges} of a complete graph.
It turns out that CONCORD-PCD simultanoeusly updates off-diagonal elements in which the associated edges are colorable with the same color.
As a result, the number of steps required for updating off-diagonal elements reduces from $p(p-1)/2$ to $p-1$ (for even $p$) or $p$ (for odd $p$), where $p$ denotes the number of variables.
We prove that the number of such steps is irreducible
In addition, CONCORD-PCD is tailored to single-instruction multiple-data (SIMD) parallelism. A numerical study shows that the SIMD-parallelized PCD algorithm implemented in graphics processing units (GPUs) boosts the CONCORD-CD algorithm multiple times. The method is available in the R package \texttt{pcdconcord}.
\end{abstract}

\section{Introduction}\label{sec:intro}


The estimation of a precision matrix, the inverse of a covariance matrix, is essential for many downstream data analyses and has wide application in social science, economics, and physics, among others. 
Directly estimating the true precision matrix under some sparsity conditions is a popular choice where the number of variables ($p$) is relatively large compared to the sample size ($n$). 
Examples include
likelihood-based \citep{Yuan2007,Friedman2008,Witten2011,Mazumder2012},
regression-based  \citep{Meinshausen2006,Peng2009,Sun2013,Khare2015}
and constrained $\ell_1$-minimization approaches \citep{Cai2011,Cai2016b,Pang2014}.
The CONvex partial CORrelation selection methoD (CONCORD)
proposed by \cite{Khare2015} is a variant of a regression approach called SPACE \citep{Peng2009}.
It has good theoretical properties: the objective function is convex and the estimator is statistically consistent (provided that the true counterpart is sparse) while satisfying the symmetry requirement.

Scalability of CONCORD and any other precision matrix estimation methods is a key challenge for application. Roughly speaking, they require at least $O(np^2)$ or $O(p^3)$ of float-point operations (``flops''). As $p$ increases, the computation time increases dramatically.
For example, a coordinate descent algorithm for the CONCORD (CONCORD-CD)
 proposed in \cite{Khare2015} requires  3440.95 (sec) for $n=2000$ and $p = 5000$
in our numerical study.  Detailed settings are introduced in Section \ref{sec:simulation}. 
Applications to high-dimensional data, such as gene regulatory analysis and portfolio optimization, 
face this computational challenge.

This study aims to fill this scalability gap by proposing a novel parallelization of the CONCORD-CD algorithm, namely, CONCORD-PCD algorithm. 
A high-level motivation of the algorithm is as follows.
Recall that the CONCORD-CD runs \emph{consecutive updates}, because the cyclic coordinate descent algorithm minimizes a target objective function with respect to one coordinate direction
at each update while the other
coordinates are fixed. Thus, each   update   requires the result of the previous update, which 
%
is essential to guarantee convergence. 
As a result, the CD algorithm for CONCORD (i.e., CONCORD-CD)  consumes 
$p(p+1)/2$ serial updates
per iteration to update the entire precision matrix. 
We observe that a careful reordering of the elements to be updated allows some consecutive updates
to  run \emph{simultaneously} even as convergence guarantee is preserved.
This is because every elements corresponding to the carefully chosen set of consecutive updates
are \emph{independent} in a sense that
an update for each element does not
require the results of the 
updates for the other elements in the given set. 

%
We systematize such observation by the lens of the \emph{edge coloring}, a well-known concept in graph theory. Edge coloring is an assignment of \emph{colors} to the edges of a graph in a way that any pair of  edges sharing at least one vertices has different colors.
Specifically, we build a conceptual bridge between  \emph{updating an element} of the off-diagonal elements in CONCORD-CD and \emph{coloring the associate edges} of a complete graph.
Then, we prove that a set of the off-digonal elements can be updated simultaneously in parallel if the associated  edges are colorable with the same color. 
This theorem enables us to  employ the so-called circle method, a scheduling principle to color a complete graph with the minimal number of colors (i.e., parallel steps).
Consequently, the consecutive steps required to update all the off-diagonal elements reduce to $p-1$ ($p$)  when $p$ is even (odd), where each step runs a simultaneous update of $p/2$ ($(p-1)/2$) elements. After then, the entire diagonal elements can be updated by one additional step.

We also provide the details to implement the CONCORD-PCD algorithm tailed for graphics processing unit (GPU) devices, which is also available in R Package \texttt{pcdconcord} 
at \url{http://sites.google.com/view/seunghwan-lee}.  
GPU devices receive  growing attention in statistical computing since GPU has many light-weight cores that can enormously reduce computation time when the given operations are adequate for
single-instruction multiple-data (SIMD) parallelism.
SIMD parallelism refers to a processing method where multiple processing units perform the same operation on multiple data points. A typical example of SIMD is summing two vectors where the sum of each element is conducted by one sub-processing unit.
We note that the CONCORD-PCD algorithm is well-suited for SIMD parallelism.
Our numerical results show that the GPU-parallelized CONCORD-PCD algorithm boosts the original CONCORD-CD algorithm implemented in the CPU multiple times.

Parallelization of coordinate descent algorithms have been considered in the literature. 
\cite{Richtarik2016a} and \cite{Bradley2011} proposed parallelized coordinate descent algorithms for regularized convex loss functions.
In particular, \cite{Richtarik2016a} randomly partitioned the coordinates and distributed the partitioned subprograms. 
\cite{Bradley2011} updated  the iterative solution by the direction of the average of increments on each axis.
It is worth noting that both studies required an appropriate 
learning rate (a constant multiplied by the descent direction) to guarantee convergence to the optima.
 In practice,  the optimal learning rate is unknown and is set sufficiently small,  which results in a large number of iterations for convergence. In contrast, our algorithm does not involve selection of the learning rate to guarantee convergence.
The literature of sparse precision matrix estimation
has considered the  parallelization  of the likelihood-based and constrained $\ell_1$-minimization approaches \citep{Hsieh2013,Hsieh2014,Wang2013d}. To the best of our knowledge, 
it has devoted much less attention to  the regression-based approach, including CONCORD.

The remainder of this paper is organized as follows. In Section \ref{sec:prelim}, we briefly review the CONCORD-CD algorithm as well as key concepts in graph theory, focusing on the edge coloring. In Section \ref{sec:method}, we provide the details of the CONCORD-PCD algorithm. In Section \ref{sec:theory}, we prove the convergence of  the CONCORD-PCD algorithm by leveraging edge coloring. In Section \ref{sec:simulation}, we demonstrate the computational gain of the CONCORD-PCD algorithm with extensive numerical studies. Finally, we conclude the paper in Section \ref{sec:final}.

\section{Preliminaries}\label{sec:prelim}

\subsection{CONCORD: the objective function and coordinate descent algorithm}

CONCORD (\citealp{Khare2015}) is a regression-based pseudo-likelihood method for sparse precision matrix estimation. 
The CONCORD estimator is given by  a minimizer  of the following convex objective function:
\begin{equation} \label{eqn:concord}
L(\Omega;\lambda)= -\sum_{i=1}^p n\log \omega_{ii} + \frac{1}{2}\sum_{i=1}^p \sum_{k=1}^n \Big( \omega_{ii} X_{ki} + \sum_{j\neq i} \omega_{ij}  X_{kj} \Big)^2 + \lambda
\sum_{i<j} |\omega_{ij}|,
\end{equation}
where $\Omega = (w_{ij})_{1 \leq i,j \leq p}$ is a precision matrix term, ${\bf X} = (X_{ki})_{1 \leq k \leq n, 1 \leq i \leq p}$ is the given data matrix (assumed to be centered columnwise), and $\lambda > 0$.
The consistency of the solution was proved when the true counterpart is sparse. 

The CONCORD-CD algorithm proposed in the paper cyclically minimizes \eqref{eqn:concord} with respect to each element.
We briefly review the algorithm for completeness. 
With a slight abuse of notation, let $(\hat{\omega}_{ij})$ be the current update of the algorithm.
First, the $p$ diagonal elements are updated by
\begin{equation} \label{eqn:diag}
\hat{\omega}_{ii}^{\rm new}
\gets \frac{- \sum_{j\neq i} \hat{\omega}_{ij} T_{ij} + \sqrt{ \big(\sum_{j\neq i} \hat{\omega}_{ij} T_{ij} \big)^2 + 4 n T_{ii}}}{2 T_{ii}}.
\end{equation}
Second, the $p(p-1)/2$ off-diagonal elements are updated by 
\begin{equation} \label{eqn:offdiag}
\hat{\omega}_{ij}^{\rm new}
\gets \frac{{\rm Soft}_{\lambda}(-\sum_{j' \neq j} \hat{\omega}_{ij'} T_{jj'}
- \sum_{i' \neq i} \hat{\omega}_{i'j} T_{ii'} )  }{T_{ii} + T_{jj}},
\end{equation}
where $T_{ij}$ is $(i,j)$th element of ${\bf X}^T {\bf X}$,
${\rm Soft}_\tau (x) = {\rm sign}(x) (|x| - \tau)_+$,
and  $(x)_+ = \max (0, x)$.

Note that each element is updated consecutively; that is, once an element is updated, it is used as input in the right-hand sides of \eqref{eqn:diag} and \eqref{eqn:offdiag}.
Thus, the CONCORD-CD algorithm appears to be inherently serial.
In Section \ref{sec:method}, we propose partitioning of the  updating equations  
for the off-diagonal updates \eqref{eqn:offdiag} such that  each partitioned group of 
updating equations  can run simultaneously in parallel. In Section \ref{sec:theory}, we prove that the convergence guarantee is preserved. Our claim will leverage the edge coloring described below.

\subsection{Undirected graph and edge coloring}

We briefly review key concepts of the edge coloring in graph theory.  See  \cite{Nakano1995} and \cite{Formanowicz2012} for comprehensive reviews.

A (simple undirected) \emph{graph} $\MC{G}$ is defined by an ordered pair of sets of \emph{nodes} and \emph{edges}, namely, $\MC{G} = \MC{G}(V, E)$. $V$ is a set of nodes (also called vertices), typically representing variables, say, $V = \{ 1, \ldots, p \}$. $E$ is a set of edges that are unordered pairs of nodes, $E \subseteq  \{ \{i, j\} ~|~ (i,j) \in V \times V, \, i \neq j \}$. 
For simplification, we denote an edge by $ij \in E$ with a slight abuse of notation. We say that the pair $i,j \in V$ is \emph{connected} if $ij \in E$.
One example of a graph is a complete graph with $p$ vertices, say, $\MC{K}_p$, in which every pair of  nodes is connected. In other words, there are $p(p-1)/2$ of edges in $\MC{K}_p$.

\emph{Edge coloring} is defined as an assignment of \emph{colors} to the edges of a graph such that
any pair of  \emph{adjacent} edges (edges sharing at least one vertices) is colored with different colors.
%
Coloring all edges with mutually distinct colors, say, $1, \ldots, K$,
where $K$ is a number of edges in $\MC{G}(V,E)$, is a typical example of edge coloring. 
The central interest is to minimize the number of colors, $K$. 
The following theorem, a special case of Baranyai's Theorem, mathematically establishes  optimal edge coloring for complete graphs.
\begin{theorem}[Baranyai's Theorem]\label{thm:edgeColoring}
Suppose that $\MC{K}_p$ is an undirected complete graph with $p$ vertices. 
the minimum number of colors that can edge-color $\MC{K}_p$ is $p-1$ (if $p$ is even) or $p$ (if $p$ is odd).
\end{theorem}
\noindent For example, Table \ref{table:edgeColoring} compares two edge-colorings for $\MC{K}_6$; the left graph  represents a trivial edge coloring with mutually distinct colors, while the right graph is an
 example of Theorem \ref{thm:edgeColoring} with a minimal number of colors.

Note that  our usage of graph is unrelated  to Gaussian graphical models, where the presence of an edge implies  nonzero partial correlation in a true precision matrix. 

\begin{table}[h]
    \centering
    \renewcommand{\arraystretch}{1.1}
    \begin{tabularx}{.95\textwidth}{ m{3.7cm}*2{|X} }    
    \hline
    Coloring scheme & with mutually distinct colors & with minimal number of colors  \\ \hline
    Graph with coloring (e.g. $\MC{K}_p$ with $p=6$) & 
    \includegraphics[width=1.5in]{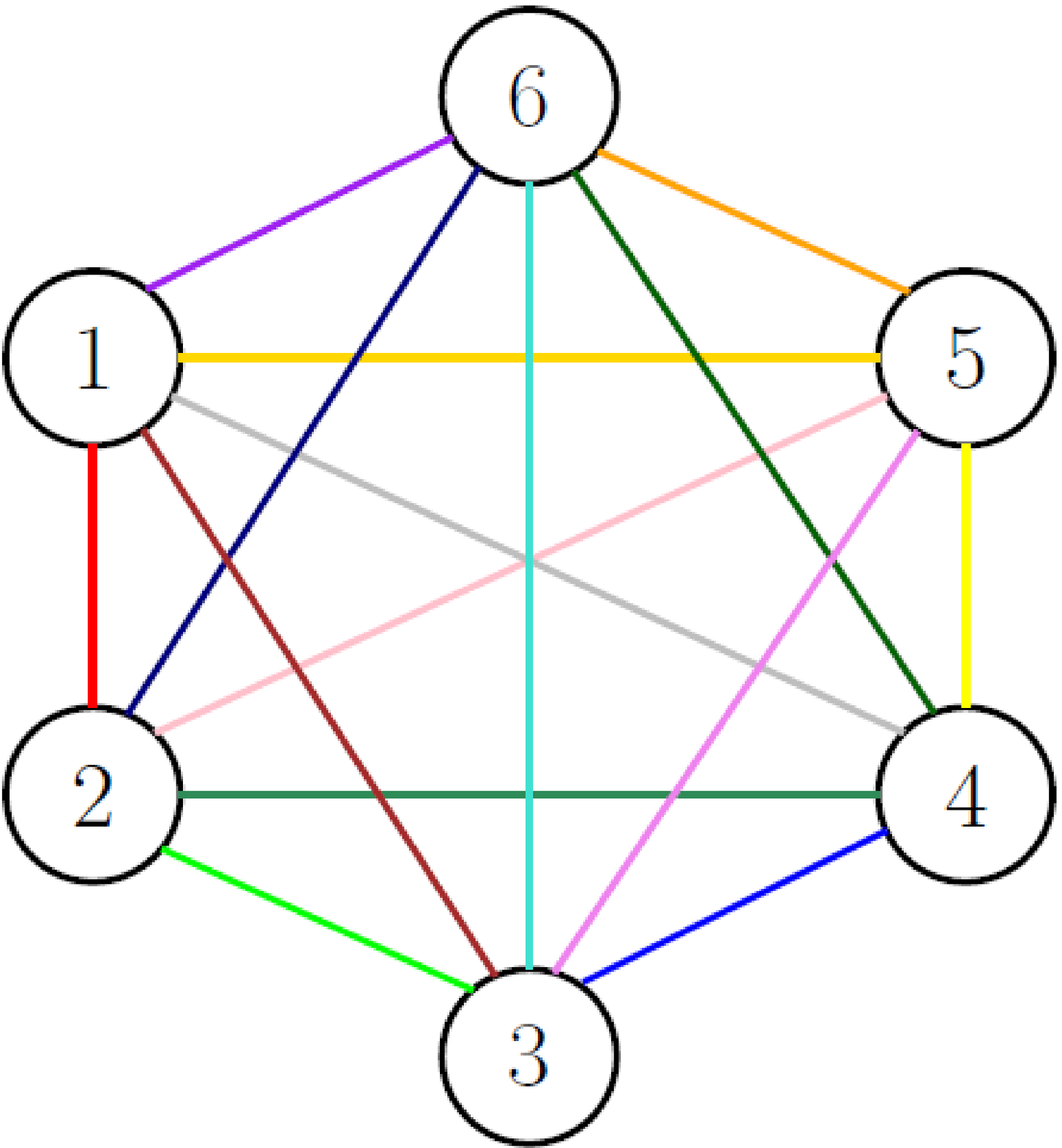} & 
    \includegraphics[width=1.5in]{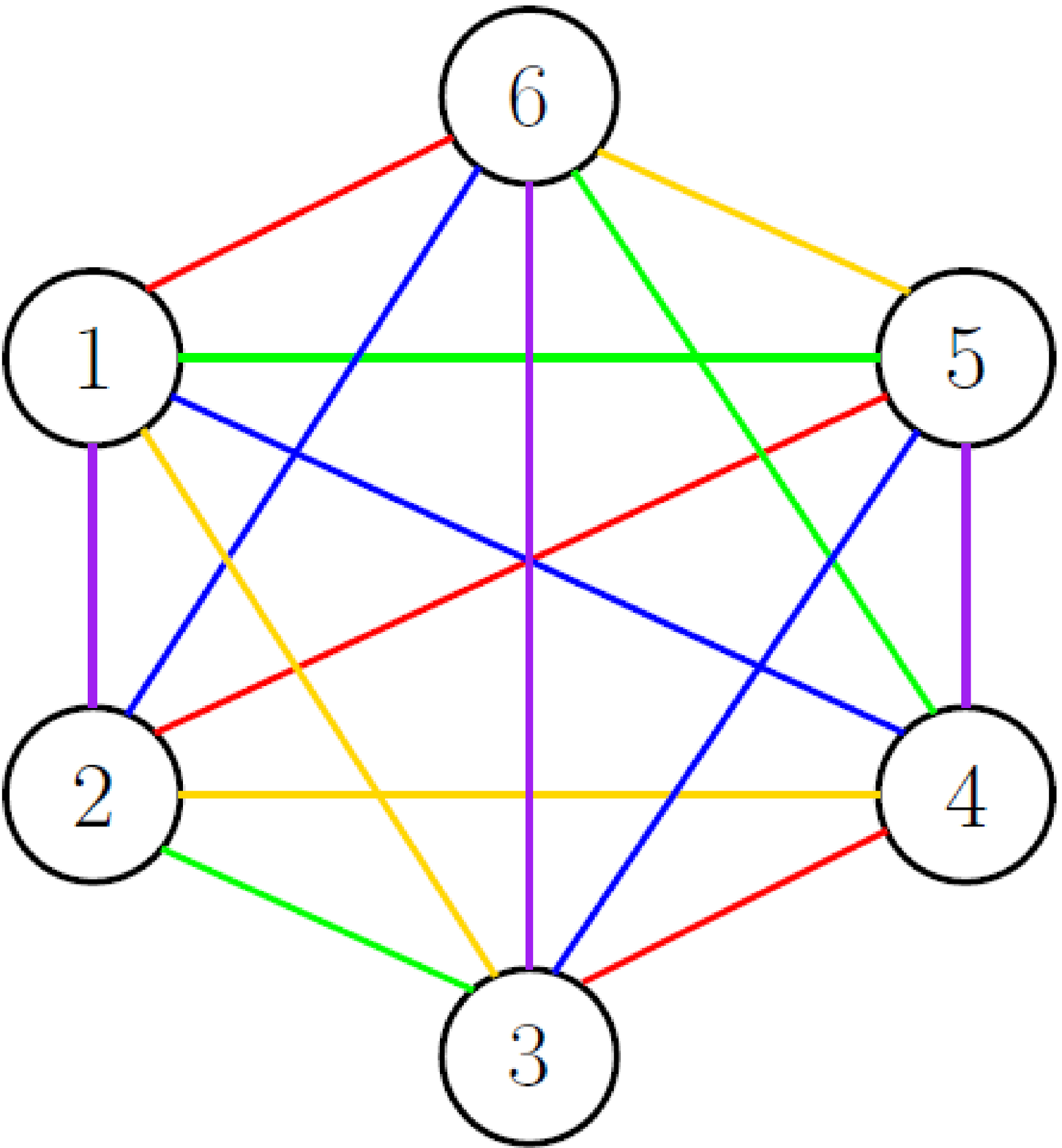}  \\ \hline
    \# of colors &  $p(p-1)/2 = 15$  & $p-1 = 5$ \\ \hline
    \# of edge(s) for each color &  $1$ &  $p/2 = 3$ \\ \hline
    Collections of edges of the same colors 
    & $\{12\}$, $\{13\}$, $\{14\}$, $\{15\}$, $\{16\}$, $\{23\}$, $\{24\}$, $\{25\}$, $\{26\}$, $\{34\}$, $\{35\}$, $\{36\}$, $\{45\}$, $\{46\}$, $\{56\}$
    & $\{16, 25, 34\}$, $\{15, 23, 46\}$, $\{14, 26, 35\}$, $\{13, 24, 56\}$, $\{12, 36, 45\}$ \\ \hline
    \end{tabularx}
    \caption{An intuitive explanation of edge coloring.}
    \label{table:edgeColoring}
\end{table}

\section{Parallel Coordinate Descent algorithm for CONCORD (CONCORD-PCD)}\label{sec:method}


In this Section, we construct the proposed algorithm and explain implementation  details. 
We begin with a motivational example.
Suppose $p=6$, and let $\hat{\Omega} = (\hat{\omega}_{ij})$ be the current iterate of the CONCORD-CD algorithm. 
From \eqref{eqn:offdiag},  the elements used to calculate $\hat{\omega}_{16}^{\rm new}$, $\hat{\omega}_{25}^{\rm new}$,  and $\hat{\omega}_{34}^{\rm new}$ can be displayed as below: 

\begin{center}
\scalebox{0.7}{
\begin{tabular}{ccc}
$\hat{\omega}_{16}^{\rm new} \gets \BPA{*6c}
    \hat{\omega}_{11}& \hat{\omega}_{12} & \hat{\omega}_{13} & \hat{\omega}_{14} & \hat{\omega}_{15} &  \\
    \hat{\omega}_{12}& & & & \times & \hat{\omega}_{26}\\
    \hat{\omega}_{13}& & & \times & & \hat{\omega}_{36}\\
    \hat{\omega}_{14}& & \times& & & \hat{\omega}_{46}\\
    \hat{\omega}_{15}& \times & & & & \hat{\omega}_{56}\\
    & \hat{\omega}_{26}& \hat{\omega}_{36} & \hat{\omega}_{46} & \hat{\omega}_{56} & 
    \hat{\omega}_{66}
\EAP
$ &
$\hat{\omega}_{25}^{\rm new} \gets \BPA{*6c}
    & \hat{\omega}_{12}& & & \hat{\omega}_{15}& \times\\
    \hat{\omega}_{12} & \hat{\omega}_{22}& \hat{\omega}_{23}& \hat{\omega}_{24} & & \hat{\omega}_{26} \\
    & \hat{\omega}_{23} & & \times& \hat{\omega}_{35}& \\
    & \hat{\omega}_{24}& \times& & \hat{\omega}_{45}&  \\
    \hat{\omega}_{15}&  & \hat{\omega}_{35} & \hat{\omega}_{45} & \hat{\omega}_{55} & \hat{\omega}_{56}  \\
    \times & \hat{\omega}_{26}& &  & \hat{\omega}_{56}& 
\EAP$
&
$\hat{\omega}_{34}^{\rm new} \gets \BPA{*6c}
    & & \hat{\omega}_{13}& \hat{\omega}_{14}& & \times\\
    & & \hat{\omega}_{23}& \hat{\omega}_{24}& \times & \\
    \hat{\omega}_{13}& \hat{\omega}_{23} & \hat{\omega}_{33} &  & \hat{\omega}_{35} & \hat{\omega}_{36} \\
    \hat{\omega}_{14} & \hat{\omega}_{24}& & \hat{\omega}_{44}& \hat{\omega}_{45} & \hat{\omega}_{46}\\
        & \times & \hat{\omega}_{35}& \hat{\omega}_{45} & & \\
   \times & &\hat{\omega}_{36} &\hat{\omega}_{46} &  & 
\EAP$

\end{tabular}
}
\end{center}

We note that the updates of the three elements considered  do not use each other; otherwise, they would have appeared at the locations indicated as ``$\times$''.
To understand the implication, suppose that ${\omega}_{16}$, ${\omega}_{25}$,  and ${\omega}_{34}$ 
 are scheduled to be consecutively updated in the CONCORD-CD algorithm. 
The algorithm runs the three updates serially with a single processing unit.
However, by the independency observed above, the actual computation of the three updates can run \emph{simultaneously} on multiple processing units sharing memory storing  $\{\hat{\omega}_{ij}\} \backslash \{ \hat{\omega}_{16}, \hat{\omega}_{25}, \hat{\omega}_{34}\}$. 
Thus, under a parallel computing environment, the three serial steps of updates can be replaced with one parallel step.
We would like to mention that the associated edges $16$, $25$, and $34$ are colored with the same color in the right part of Table \ref{table:edgeColoring}.
In fact, we can show that any collection of $\hat{\omega}_{ij}$, with the associated edges assigned the  same color, can be updated simultaneously if they are consecutively updated in the CONCORD-CD algorithm. 
In this example, $p(p-1)/2=15$ of serial steps of updates can be replaced with $p-1=5$ steps with the aid of multiple processing units.

The following subsections generalize the motivation. In Section \ref{subsec:analogy},  we propose an analogy between the edge coloring of $\MC{K}_p$ and the scheduling of off-diagonal updates in the CONCORD-CD algorithm. 
In Section \ref{subsec:circle},
we employ the circle method, a particular scheme for edge-coloring $\MC{K}_p$, to explain the proposed parallelization of the CONCORD-CD algorithm. We hereafter refer to the proposed algorithm as CONCORD-PCD.
In Section \ref{subsec:algo}, we describe the complete algorithm and provide the implementation details.
The theoretical guarantees are deferred to  Section \ref{sec:theory}.

\subsection{Analogy between edge coloring and update ordering}\label{subsec:analogy}

 We now assosiate vertex $r$ of the complete graph $\kappa_p$ with the $r$-th variable and then edge $ij$ with $\omega_{ij}$ of the given data. We propose the following analogy:
\begin{center}
(A) Associate the edge-coloring of edge $ij$ by color $k$ \\ with the update of $\hat{\omega}_{ij}$ as in \eqref{eqn:offdiag} at the $k$-th step.
\end{center}
For example, coloring all edges with colors 1  through $p(p-1)/2$ is a trivial edge-coloring of $\MC{K}_p$. By (A), this coloring scheme is associated with the  original CONCORD-CD algorithm: all the coordinate descent updates of the off-diagonal elements run serially.
On the other hand, coloring multiple edges $ij$ with the same $k$-th color means that the associated $\omega_{ij}$'s  are simultaneously updated given the same current iterate.
In Section \ref{sec:theory}, we will show the well-definedness of (A), i.e., any set of edges colorable with the same color can  be updated simultaneously.
%

\subsection{The circle method of edge-coloring $\MC{K}_p$}\label{subsec:circle}

The circle method is used to assign colors to the edges of $\MC{K}_p$ with minimal number of colors. See \cite{Dinitz2006} for a comprehensive review.
By (A), application of the circle method implies that $p/2$ elements can be updated simultaneously, and 
$(p-1)$ stpes (i.e., colors) are required to update all off-diagonal elements if $p$ is even.
Where $p$ is odd, $(p-1)/2$ off-diagonal elements can be updated
simultaneously with $p$ steps.

Here, we provide a sketch of the circle method.  Its implementation details in Algorithm  \ref{algo:bcd1}.
 We define  a variable $p_{even}$ as
$p_{even}=p$ if $p$ is even and $p_{even}=p+1$ if $p$ is odd to handle the differences between the two situations.
The circle method of CONCORD-PCD consists of following steps:
\begin{itemize}
    \item[(i)] Clockwisely rotate  the round-robin table with the $(1,1)$ element is fixed, which results in $p_{even}-1$ distinct tables:
\end{itemize}    
\begin{center}
\scalebox{0.55}{
    \begin{tabular}{*{17}{|c}|} \cline{1-5} \cline{7-11} \cline{13-17}
    $1$ & $2$  & $3$ & $\cdots$ & ${p_{even}/2}$& \multirow{2}{*}{$\rightarrow$} & 
    $1$ & ${p_{even}}$ & $2$ & $\cdots$ & ${p_{even}/2-1}$&
    \multirow{2}{*}{$\rightarrow \cdots \rightarrow$} &
    $1$ &   $3$ &   $4$ &  $\cdots$ &  ${p_{even}/2 +1}$  \\ \cline{1-5} \cline{7-11} \cline{13-17}
    ${p_{even}}$ & ${p_{even}-1}$  & ${p_{even}-2}$ & $\cdots$ & ${p_{even}/2+1}$&  & 
    ${p_{even}-1}$ & ${p_{even}-2}$ & ${p_{even}-3}$ & $\cdots$ & ${p_{even}/2}$ & &
    $2$ &  ${p_{even}}$ &  ${p_{even}-1}$ &  $\cdots$ &   $p_{even}/2+2$ \\ \cline{1-5} \cline{7-11} \cline{13-17}
    \end{tabular}
    }
\end{center}

\begin{itemize}
\item[(ii)] Define target sets: We call a pair of two indices in the same column as a matching pair. 
We define the $k$-th target set, $I_k$, as the collection of all matching pairs in the $k$-th table  in (i), $k=1,\ldots,p_{even}-1$. For example, the $k$-th target set in the first table in (i)  is $I_k = \{ \{1, p_{even}\}, \{2, p_{even}-1\}, \ldots, \{p_{even}/2, p_{even}/2+1 \} \}$. 

\item[(iii)] Discard a pair containing the $(p+1)$ index in $I_k$, $k=1,\ldots,p_{even}-1$
if $p$ is odd. 

\item[(iv)] Color $I_k$ (the edges associated with $I_k$) by the $k$-th color, $k=1,\ldots,p_{even}-1$. In other words, update the off-diagonal elements associated to $I_k$ simultaneously at the $k$-th parallel step.
\end{itemize}

Consequently, we update the off-diagonal elements of $\hat{\Omega}$ in $p_{even} -1$ steps. 
Note that the pair in (iii) is implicitly discarded in the implemented circle method, because
we can skip the pair
containing the $(p+1)$-th index when updating the off-diagonal elements.
This circle method  applies regardless of whether $p$ is even or odd since
the numbers of pairs and iterations are $(p/2, p-1)$ where $p$ is even
and $((p+1)/2-1, (p+1)-1) = ((p-1)/2, p)$ where $p$ is odd, in which case
a pair is discarded and the number of pairs to be simultaneously updated 
is computed by $p_{even}/2$ (i.e., $p_{even}/2 - 1$).

\subsection{A complete algorithm and implementation details}\label{subsec:algo}

A complete CONCORD-PCD algorithm is described in  Algorithm \ref{algo:bcd1}.
The inner loop of the complete CONCORD-PCD algorithm consists of two parallel update procedures
for off-diagonal elements and diagonal elements.
As described in the previous section, 
the parallel update of off-diagonal elements involves
$p_{even}-1$ steps of updating $p_{even}/2$ elements in parallel.
In addition, the parallel update of diagonal elements involves one step
since all $p$ diagonal elements can be updated simultaneously
with the given off-diagonal elements.
Thus, the complete algorithm runs $p_{even}$ steps per one outer iteration. The algorithm converges to a global minima, which is proved in Theorem \ref{thm:algoconv} in Section \ref{sec:theory}.

To further accelerate CONCORD-PCD, we also apply the cyclic reduction technique for pairwise comparison
to calculate $|\hat{\Omega}^{(k)}- \hat{\Omega}|_\infty$,
where $|A|_\infty = \max_{i,j} |A_{ij}|$ is the maximum absolute value of matrix $A$. 
Let $\hat{\theta} = (\hat{\theta}_1,\ldots, \hat{\theta}_m)={\rm vech}(\hat{\Omega})$,
which is a half-vectorization for the parameter estimate $\hat{\Omega}$,
and ${\bf d} = (d_j)_{1\le j \le m} = \hat{\theta}^{new}-\hat{\theta}^{old}$.
We further let $z = \lceil \log_2(m) \rceil$, where $\lceil x \rceil$
is the smallest integer greater than or equal to $x$.
Consider a calculation of $\|{\bf d}\|_\infty$,
where $\|{\bf d}\|_\infty = \max_j |d_j|$ is the $L_\infty$-norm for vector ${\bf d}$.
The pairwise comparison in the proposed algorithm is conducted as follows:
\begin{itemize}
    \item Initialization: for $q=z-1$,

$d_j \gets \max(|d_j|, |d_{j+2^q}|)$ if $j+2^q\le m$
and $d_j \gets d_j$ if $j+2^q > m$ for $j=1,\ldots, 2^q$,

   \item Cyclic reduction: for $q=z-2, \ldots, 0$,  
   
$d_j \gets \max(|d_j|, |d_{j+2^q}|)$  for $j=1,\ldots, 2^q$.
\end{itemize}

After the cyclic reduction step for $q=0$, 
the first element $d_1$ of ${\bf d}$ becomes equal to $\|{\bf d}\|_\infty$.
With GPU-parallel computation, we can simultaneously compare
$2^q$ pairs for each step in the cyclic reduction,
and then the computational cost can be 
reduced as $O(\log_2(m))$ if $2^{z-1}$ CUDA cores are available.

\begin{algorithm}[h]
\caption{Parallel coordinate descent algorithm for CONCORD (CONCORD-PCD)} \label{algo:bcd1}
\begin{algorithmic}[1] 
\Require Data matrix ${\bf X}$ of size $n$ by $p$, $\hat{\Omega}^{(0)}=(\hat{\omega}_{ij}^{(0)})$, $\lambda$, and $\delta_{tol}$

    \State $t \gets 0$, $\hat{\Omega} \gets \Omega^{(0)}$, $T \gets {\bf X}^T {\bf X}$, $p_{even} \gets p$ \Comment{initialization}
    \If{$p$ is odd}
     \State{$p_{even} \gets p+1$}
    \EndIf
    \State{ $(j_1, \ldots, j_{p_{even}})  \gets (1, \ldots, p_{even})$}
    \Comment{initialization of index set}
     \Repeat 
       \State $t \gets t+1$
       \For{$k=1,2,\ldots, p_{even}-1$}
       \Comment{updating off-diagonal elements}
            \State Define a target set $I = \{ (r,s) ~|~ r = j_q, \, s = j_{p_{even}-q+1}, \, q=1,2, \ldots, p_{even}/2 \}$
            \State Update, for all $(r,s) \in I$ such that $r,s\neq p+1$,
            \Comment computed in parallel
                \[ \hat{\omega}_{rs}
\gets \frac{{\rm Soft}_{\lambda}(-\sum_{u \neq s} \hat{\omega}_{ru} T_{su}
- \sum_{u \neq r} \hat{\omega}_{us} T_{ru} )  }{T_{rr} + T_{ss}} \] 
            \State tmp $\gets j_{p_{even}}$, $(j_3, \ldots, j_{p_{even}}) \gets (j_2, \ldots, j_{p_{even}-1})$, $j_2 \gets$ tmp
       \EndFor 
       \For{$k=1,2,\ldots,p$} 
       \Comment{updating diagonal elements in parallel}
       	\State \[ \hat{\omega}_{ii}
\gets \frac{- \sum_{j\neq i} \hat{\omega}_{ij} T_{ij} + \sqrt{ \big(\sum_{j\neq i} \hat{\omega}_{ij} T_{ij} \big)^2 + 4 n T_{ii}}}{2 T_{ii}} \] 
       \EndFor
       \State {$\delta \gets | \hat{\Omega}^{(t)} - \hat{\Omega} |_\infty$}
       \Comment{computed by cyclic reduction}
       \State {$\hat{\Omega}^{(t)} \gets \hat{\Omega}$}
   \Until{ $\delta< \delta_{tol}$}      
\end{algorithmic}
\end{algorithm}

\section{Properties}\label{sec:theory}

In this Section, we prove computational properties of CONCORD-PCD algorithm.

Recall the motivating example in Section \ref{sec:method} in which $\hat{\omega}_{16}$, $\hat{\omega}_{25}$ and $\hat{\omega}_{34}$ are simultaneously updateable in the sense that
their updates do not require each other's current iterates.
The following lemma characterizes   a sufficient condition for independent updates. 
%
%
\begin{lemma}\label{lemma:PCD1} Suppose that two edges $\{i,j\}$ and $\{k,l\}$ of $\MC{K}_p$ are colorable by the same color.
Then, the updates of $\hat{\omega}_{ij}$ and $\hat{\omega}_{kl}$ by the CONCORD-CD algorithm does not contain each other.
\end{lemma}
\begin{proof}
For an edge $\{i, j\}$, we define $U(\{i, j\})$ as the family of coordinates
needed to update $\omega_{ij}$  by \eqref{eqn:offdiag}.
From the two summation operations in the right-hand side of \eqref{eqn:offdiag}, we have $U(\{i, j\}) = \tilde{U}(i,j) \cup \tilde{U}(j,i)$, where $\tilde{U}(i,j)$ is defined as
\[
\tilde{U}(i, j) := \{ (i, i') \,:\, i' \neq j, ~ 1 \leq i' \leq p \}, ~~~  \mbox{for}~ 1 \leq i, j \leq p ~ \mbox{and} ~ i \neq j.
\]
By the definition of edge coloring, if two edges are colorable by the same color, then they do not
 share vertices, i.e., $i,j,k,l$ are distinct integers.
Observe that $k \neq i$ and $l \neq i$ imply $(k,l) \notin U(i,j)$ and $(l,k) \notin U(i,j)$. Similarly, by $k \neq j$ and $l \neq j$, we have $(k,l) \notin U(j,i)$ and $(l,k) \notin U(j,i)$. Combining these leads to $(k, l) \notin U(\{i, j\})$ and $(l, k) \notin U(\{i, j\})$. Hence, $\omega_{kl}$ is not used for updating $\omega_{ij}$. 
In contrast,  we can verify that $\omega_{ij}$ is not used for the update of $\omega_{kl}$ by interchanging  the role of subscripts.
\end{proof}
\begin{lemma}\label{lemma:PCD3}
Suppose that any collection of edges of $\MC{K}_p$, say,  $\{i_1 j_1, \ldots, i_q j_q\}$, is colorable with the same color.
Then, the associated elements in $\hat{\Omega}$, that is, $\hat{\omega}_{i_1 j_1}$ through $\hat{\omega}_{i_q j_q}$,are simultaneously updatable by the CONCORD-PCD algorithm.
\end{lemma}
 Lemma \ref{lemma:PCD3} straightforward from Lemma \ref{lemma:PCD1}. The Lemmas provides a characterization for the motivating example as well as Table \ref{table:edgeColoring}:  the sufficient condition for the simultaneous updatability of  $\hat{\omega}_{16}$, $\hat{\omega}_{25}$ and $\hat{\omega}_{34}$ is from the observation that the edges $16$, $25$, and $34$ are colorable with the same color.

Using Lemma \ref{lemma:PCD3}, we can show the global convergence property of the proposed algorithm.
\begin{theorem}\label{thm:algoconv} 
Algorithm \ref{algo:bcd1} converges to the minimizer of \eqref{eqn:concord}. 
\end{theorem}
\begin{proof}
We will show that the updates of Algorithm \ref{algo:bcd1} are essentially the serial reordering of the CONCORD-CD algorithm.
To fix the idea, assume that $p$ is even (extending to odd $p$ is straightforward).
We further  fix one outer loop at line 7  of  Algorithm \ref{algo:bcd1}.
For the inner parallel step $k$, $k=1, \ldots, p-1$, let  $I$ be the target set defined at line 9, which coincides the $k$-th target set $I_k$ in Section \ref{subsec:circle}.
Let $J = \{(1,1), \ldots, (p,p)\}$ denote the indices for the main diagonal.
Then, the update order of the indices of $\Omega$ given the algorithm  is
\[
\mbox{\sf U1}: ~~~~~~~ I_1 \to I_2 \to \cdots \to I_{p-1} \to J, 
\]
where the elements associated with  each set is calculated  simultaneously.
Now, consider a serialized  update of {\sf U1}, say {\sf U2}, which inherits the order in {\sf U1},
and  the elements in each $I_k$ and $J$ are arbitrarily ordered.
We can apply Lemma \ref{lemma:PCD3} to inductively verify that {\sf U1} and {\sf U2} produce exactly the same updated $\hat{\Omega}$.
Now recall that $I_1, \ldots, I_{p-1}$, and $J$ in {\sf U1} are a disjoint union for all coordinates $\{ (i,j): 1 \leq i,j \leq p\}$.
The serialized  update scheme {\sf U2} then satisfies the conditions of Theorem 5.1 in \cite{Tseng2001}, which guarantees that convergence  to the global minima.
Thus, iterating {\sf U1} also converges to the global minima, which completes the proof.
\end{proof}

The construction of {\sf U1} in the proof can easily be extended to arbitrary edge coloring of $\MC{K}_p$. Specifically, given an edge coloring of $\MC{K}_p$ with colors $1, 2, \ldots, C$, one can mimic the proof
to organize a parallelizable update order of CONCORD-CD algorithm with $C$ steps 
for the off-diagonal elements plus $1$ step for the diagonal elements.
One would naturally want to know how much we can reduce the number $C$ while preserving convergence, considering that the fewer the steps we need to follow, the more we can maximize the utility of parallel processing units. 
We note that the number of parallel steps for the off-diagonal update in Algorithm \ref{algo:bcd1} is \emph{minimal}. This is due to the construction of our edge coloring with $p-1$ (for even $p$) or $p$ (for odd $p$) colors, which is guaranteed as the minimal possible number of edge colors by Theorem \ref{thm:edgeColoring}.

\section{Numerical Study}\label{sec:simulation}

To illustrate the computational advantage of the proposed parallelization implemented on a GPU, we compare the computation time of the CONCORD-CD algorithm of \cite{Khare2015} and the proposed CONCORD-PCD algorithm. 
We developed an R package
\texttt{pcdconcord} where the CONCORD-PCD algorithm is implemented
with a dynamic library using CUDA C, which is available at
\blue{\url{https://sites.google.com/view/seunghwan-lee/software}}.
We refer to  CONCORD-PCD as ``PCD-GPU'' in the comparison to emphasize that the proposed algorithm is running on GPUs.
Next, the CONCORD-CD algorithm is available in R package \texttt{gconcord} and implemented
with a dynamic library using C with BLAS (basic linear algebra subroutine) \citep{Lawson1979}. We describe the CONCORD-CD implemented in \texttt{gconcord} as ``CD-BLAS''. 
In addition to two main algorithms (CD-BLAS and PCD-GPU), we also implemented a CONCORD-CD without BLAS, ``CD-NAIVE'', and CONCORD-PCD 
without computation on GPUs, ``PCD-CPU'',  to study the gain from
GPU parallelization.
We remark that the single precision (32-bit floating point representation)
is more efficient than the double precision (64-bit floating point representation) for the computations on GPUs.
However, the R platform only supports the double precision.
To maximize the efficiency of the GPU in the R environment, we first convert  the double-precision data  in the host (CPU) memory to single-precision data in the device (GPU) memory.
It is worth noting that Python is favorable
for CONCORD-PCD since it supports both single and double
precision for CUDA C. Thus, Python can fully utilize the computation capacity of GPUs with single precision. 
The computation time is measured in seconds on
a workstation (Intel Xeon(R) W-2175 CPU (2.50GHz) and 128 GB RAM 
with NVIDIA GeForce GTX 1080 Ti).
Note that the CONCORD-CD and CONCORD-PCD algorithms
 should produce the same estimates
after convergence since the only difference between the two algorithms is
the updating order of the matrix elements.  
In practice,  small differences might be observed due to numerical errors when the convergence tolerance $\delta_{tol}$ is not sufficiently small.

We used simulated data for the comparison. To be specific, we generate
10 data sets from a multivariate normal distribution $N_p({\bf 0}, \Omega^{-1})$ by varying the sample size ($n=500, 1000, 2000$)  and
 number of variables ($p=500, 1000, 2500, 5000$).
Because the true precision matrix affects the number of iterations for convergence of the estimator,
we also consider AR(2) and scale-free network structures for a true precision matrix, $\Omega$, 
from the literature for sparse precision matrix estimation \citep{Yuan2007, Peng2009}. 
Let $\Omega^{AR}$ and $\Omega^{SC}$,
be precision matrices for the AR(2) and scale-free networks, 
respectively.  For the  AR(2) network, the precision matrix
$\Omega^{AR}=(\omega_{ij}^{AR})_{1\le i,j \le p}$ is defined by
\begin{equation} \nonumber
\omega_{ij}^{AR} =\omega_{ji}^{AR} = \left\{
\begin{array}{lcl}
0.45 & & \mbox{ for } i=1,2,\ldots,p-1,~ j=i+1\\
0.4 & & \mbox{ for } i=1,2,\ldots,p-2,~ j=i+2\\
0  &  & \mbox{ otherwise}
\end{array} \right.
\end{equation}
For scale-free network, the precision matrix
$\Omega^{SC}=(\omega_{ij}^{SC})_{1\le i,j \le p}$ is defined by
the following steps:
\begin{itemize}
\item[](i) Generate a scale-free network $\MC{G}=\MC{G}(V,E)$
according to  Barab\'asi and Albert model \citep{Barabasi1999},
where the degree distribution $P(k)$ of $\MC{G}$ follows the power-law
distribution $P(k) \propto k^{-\alpha}$. We set $\alpha = 2.3$ following \cite{Peng2009},
which is close to the estimate from the real-world
network \citep{Newman2003};

\item[](ii) Generate a random matrix $\tilde{\Omega} = (\tilde{\omega}_{ij})$ by

$\tilde{\omega}_{ij}=\tilde{\omega}_{ji} \sim 
\mbox{Unif}\big([-1,-0.5]\cup [0.5,1]\big)$ for $\{ i,j \} \in E$, $\tilde{\omega}_{ii} = 1$ for $i=1,2,\ldots, p$;

\item[](iii) Scaling off-diagonal elements: 
$\tilde{\omega}_{ij} \gets \tilde{\omega}_{ij}/\big(1.25 \sum_{j\neq i} \tilde{\omega}_{ij}\big)$;

\item[](iv) Symmetrization:
$\Omega^{SC} \gets (\tilde{\Omega} + \tilde{\Omega}^T)/2$.
\end{itemize}

\noindent To avoid nonzero elements  of $\Omega^{SC}$ with
small magnitude,
we set $\omega_{ij}^{SC} \gets 0.1 \cdot {\rm sign}(\omega_{ij}^{SC})$
if $|\omega_{ij}^{SC}| < 0.1$ for $(i,j) \in E$.

In addition, we consider $\lambda = 0.1$ and $\lambda = 0.3$ for
the tuning parameter to evaluate
the performance at different sparsity levels of
the estimate.
Note that we did not search the optimal tuning
parameter for CONCORD since our numerical
studies aim at evaluating computational
gains. 
We set tolerance level as $\delta_{tol} = 10^{-5}$ for the convergence criteria.

Tables \ref{tb:ar2} and \ref{tb:sf} report
the averaged elapsed times for computing 
CD-BLAS, CD-NAIVE, PCD-CPU, and PCD-GPU
for the AR(2) and Scale-free networks, respectively.
We also summarize 
the averages of the number of iterations and
estimated edges of the CD and PCD
algorithms in the same tables to verify  that the proposed and original
algorithms achieve the same solution.

From Tables \ref{tb:ar2} and \ref{tb:sf},
we first observe that 
PCD-GPU is always faster than PCD-CPU
for all cases we considered.
The GPU-parallel computation
is efficient to the CONCORD-PCD algorithm
and plays a key role.
In addition, 
the efficiency of the GPU-parallelization increases
with  the number of variables.
For example, PCD-GPU is 3.08--3.95 times faster
than PCD-CPU for $p=500$, but PCD-GPU is 9.93--10.62
times faster than PCD-CPU for $p=5000$.
Such an increase in efficiency seems natural,
since the CONCORD-PCD 
simultaneously updates $p_{even}/2$ elements.

Next, we see that 
PCD-CPU is slightly slower than CD-NAIVE.
This is due to the fact that the PCD-CPU has an 
additional procedure for 
reordering the elements to be updated (line 9 in Algorithm \ref{algo:bcd1}).
Since the computation time for  CD-NAIVE and PCD-CPU is similar, 
we can conclude that PCD-GPU 
is more efficient than CD-NAIVE as well.

Finally, we compare PCD-GPU and CD-BLAS in 
the original implementation of CONCORD-CD (\texttt{gconcord}), where 
PCD-GPU was more efficient than CD-BLAS
for all cases except $(n,p)=(500,5000)$.
Specifically, PCD-GPU is 1.41 and 6.63 times faster
than CD-BLAS for the worst and the best
cases, respectively.
The efficiency gain grows with an increase in both $n$ and $p$.
For $(n,p)=(500,5000)$,
CD-BLAS is only 1.03--1.19 times
faster than PCD-GPU.

Note that the efficiency of CD-BLAS depends largely on the efficiency of BLAS (implemented by FORTRAN), as is evident from a comparison between CD-BLAS and CD-NAIVE.
For a more precise comparison, 
we replicate Tables \ref{tb:ar2} and \ref{tb:sf} in 
Figures \ref{fig:ar2} and \ref{fig:sf}, respectively.
The figures suggest that
CD-BLAS is more sensitive to
the sample size compared
to PCD-GPU.
In the AR(2) network, for example, the computation time per iteration
is measured as 0.4886 for $(n,p)=(500,1000)$ and 0.8486 for $(n,p)=(2000,1000)$
with CD-BLAS, but as 0.1817 for $(n,p)=(500,1000)$ and 0.1831 for $(n,p)=(2000,1000)$
with PCD-GPU.
This is because 
the incremental computational burden
associated with the sample size
is less for each GPU compared to the CPU because a GPU device has many CUDA cores.
For example, the GPU device NVIDIA GeForce GTX 1080 Ti used in the numerical studies has 3584 CUDA cores.

In addition, we compared the computation times of the graphical Lasso (GLASSO), which is a popular method in the likelihood approach \citep{Friedman2008},
and the constrained $\ell_1$-minimization for the inverse of matrix estimation (CLIME),
which is the constrained $\ell_1$-minimization approach \citep{Cai2011}, with ours. 
For the GLASSO, we used the R package \texttt{glasso} that boosts the original algorithm of \cite{Friedman2008} by adopting block diagonal screening rule \citep{Witten2011}.
For the CLIME, the original algorithm becomes inefficient when $p$ is large.
We apply the FASTCLIME algorithm implemented in R package \texttt{fastclime} \citet{Pang2014},
which is more efficient and uses the parametric simplex method to obtain the whole solution path of the CLIME.
Since solving the problem of the FASTCLIME is still expensive when $p$ is large, we focus on the cases
of $n=500, 1000$, $p=500, 1000$ and
$\lambda = 0.3$ for the CONCORD.
We choose the tuning parameter $\lambda$s of the GLASSO and the CLIME
by searching values that obtain similar sparsity level to that of the CONCORD with
$\lambda=0.3$, because the 
estimators of the GLASSO and CLIME are different to that of the CONCORD. 
Table 
\ref{tb:add}
reports the averages of the computation times and the number
of estimated edges.
We found that the proposed PCD-GPU 
was fastest for AR(2) and
the second-best for the scale-free network. For the scale-free network,
the efficiency of the proposed PCD-GPU 
was comparable to that of the GLASSO because
the differences in the computation times only lie between 0.24 and 1.01.
It 
has been numerically shown that the CONCORD has better performance than
the GLASSO
for identifying the non-zero elements of the precision matrix in \cite{Khare2015}.

To summarize,  
we conclude from the our numerical studies that
the proposed CONCORD-PCD is adequate for
GPU-parallel computation, and more 
efficient than CONCORD-CD when
either the number of variables or the sample size is large.
It is also noteworthy that  we implemented the PCD algorithm
with GPUs by using \texttt{cuBLAS} libary (PCD-GPU-cuBLAS),
but we found that the PCD-GPU-cuBLAS
was less efficient than the PCD-GPU implemented
by our own CUDA kernel functions. Therefore, we have omitted the  PCD-GPU-cuBLAS results.

\begin{table}[!htb]
\begin{centering} 
\caption{Average
computation time (in seconds),  number of iterations, 
and number of estimated edges for the AR(2) network.
Numbers within parentheses denote standard errors.} \label{tb:ar2} \medskip
\scalebox{0.75}{
\begin{tabular}{ccccccccccc}
\hline
\multirow{2}{*}{$\lambda$} & \multirow{2}{*}{$n$} & \multirow{2}{*}{$p$} & \multicolumn{4}{c}{Computation time (sec.)} & \multicolumn{2}{c}{Iteration} & \multicolumn{2}{c}{$|\hat{E}|$} \\ \cline{4-11} 
 &  &  & CD-BLAS & CD-NAIVE & PCD-CPU & PCD-GPU & CD & PCD & CD & PCD \\ \hline
\multirow{24}{*}{0.1} & \multirow{8}{*}{500} & \multirow{2}{*}{500} & 3.22 & 3.18 & 3.31 & 0.89 & 26.40 & 26.10 & 1976.70 & 1976.70 \\
 &  &  & (0.02) & (0.02) & (0.02) & (0.01) & (0.16) & (0.18) & (9.52) & (9.52) \\ \cline{3-11} 
 &  & \multirow{2}{*}{1000} & 13.09 & 26.28 & 28.66 & 4.87 & 26.90 & 26.80 & 5078.90 & 5078.90 \\
 &  &  & (0.09) & (0.17) & (0.18) & (0.04) & (0.18) & (0.20) & (14.19) & (14.19) \\ \cline{3-11} 
 &  & \multirow{2}{*}{2500} & 86.01 & 451.03 & 513.18 & 56.72 & 27.30 & 27.10 & 20327.30 & 20327.50 \\
 &  &  & (0.83) & (4.30) & (3.28) & (0.38) & (0.26) & (0.18) & (45.87) & (45.90) \\ \cline{3-11} 
 &  & \multirow{2}{*}{5000} & 378.04 & 3646.43 & 4149.79 & 404.75 & 27.70 & 27.30 & 64307.20 & 64307.40 \\
 &  &  & (2.94) & (19.49) & (51.93) & (2.29) & (0.15) & (0.15) & (69.26) & (69.13) \\ \cline{2-11} 
 & \multirow{8}{*}{1000} & \multirow{2}{*}{500} & 2.07 & 3.16 & 3.29 & 0.88 & 25.70 & 25.50 & 1407.20 & 1407.20 \\
 &  &  & (0.01) & (0.02) & (0.02) & (0.01) & (0.15) & (0.17) & (5.05) & (5.05) \\ \cline{3-11} 
 &  & \multirow{2}{*}{1000} & 25.90 & 25.84 & 28.27 & 4.76 & 26.20 & 26.10 & 2825.20 & 2825.20 \\
 &  &  & (0.14) & (0.13) & (0.22) & (0.03) & (0.13) & (0.18) & (4.50) & (4.50) \\ \cline{3-11} 
 &  & \multirow{2}{*}{2500} & 167.90 & 428.93 & 490.24 & 54.91 & 26.20 & 26.20 & 7216.80 & 7216.80 \\
 &  &  & (1.60) & (3.38) & (4.18) & (0.28) & (0.13) & (0.13) & (6.92) & (6.92) \\ \cline{3-11} 
 &  & \multirow{2}{*}{5000} & 694.62 & 3419.50 & 3862.95 & 385.66 & 26.20 & 26.00 & 14806.20 & 14806.20 \\
 &  &  & (3.87) & (17.43) & (17.96) & (0.05) & (0.13) & (0.00) & (14.28) & (14.28) \\ \cline{2-11} 
 & \multirow{8}{*}{2000} & \multirow{2}{*}{500} & 2.12 & 3.19 & 3.36 & 0.85 & 25.00 & 25.10 & 1393.10 & 1393.10 \\
 &  &  & (0.00) & (0.00) & (0.01) & (0.00) & (0.00) & (0.10) & (3.74) & (3.74) \\ \cline{3-11} 
 &  & \multirow{2}{*}{1000} & 21.81 & 25.76 & 27.88 & 4.65 & 25.70 & 25.40 & 2803.10 & 2803.10 \\
 &  &  & (0.39) & (0.16) & (0.26) & (0.03) & (0.15) & (0.16) & (4.70) & (4.70) \\ \cline{3-11} 
 &  & \multirow{2}{*}{2500} & 342.84 & 433.35 & 495.12 & 54.54 & 25.90 & 25.90 & 7006.90 & 7006.90 \\
 &  &  & (1.26) & (1.65) & (1.93) & (0.21) & (0.10) & (0.10) & (9.79) & (9.79) \\ \cline{3-11} 
 &  & \multirow{2}{*}{5000} & 1389.95 & 3440.95 & 3929.09 & 386.23 & 26.00 & 26.00 & 14009.30 & 14009.40 \\
 &  &  & (6.07) & (13.14) & (11.22) & (0.12) & (0.00) & (0.00) & (9.38) & (9.35) \\ \hline
\multirow{24}{*}{0.3} & \multirow{8}{*}{500} & \multirow{2}{*}{500} & 1.69 & 1.70 & 1.73 & 0.50 & 13.90 & 13.40 & 859.50 & 859.50 \\
 &  &  & (0.06) & (0.05) & (0.06) & (0.02) & (0.46) & (0.52) & (5.00) & (5.00) \\ \cline{3-11} 
 &  & \multirow{2}{*}{1000} & 7.00 & 14.19 & 14.80 & 2.55 & 14.40 & 13.70 & 1722.20 & 1722.20 \\
 &  &  & (0.13) & (0.26) & (0.45) & (0.07) & (0.27) & (0.40) & (5.87) & (5.87) \\ \cline{3-11} 
 &  & \multirow{2}{*}{2500} & 45.52 & 240.36 & 267.74 & 29.61 & 14.50 & 14.10 & 4297.80 & 4297.80 \\
 &  &  & (0.72) & (3.57) & (3.37) & (0.38) & (0.22) & (0.18) & (7.12) & (7.12) \\ \cline{3-11} 
 &  & \multirow{2}{*}{5000} & 210.20 & 1937.84 & 2289.35 & 215.59 & 14.80 & 14.50 & 8624.10 & 8624.10 \\
 &  &  & (4.05) & (38.32) & (23.56) & (2.47) & (0.29) & (0.17) & (17.08) & (17.08) \\ \cline{2-11} 
 & \multirow{8}{*}{1000} & \multirow{2}{*}{500} & 1.06 & 1.59 & 1.62 & 0.46 & 12.40 & 12.10 & 853.60 & 853.60 \\
 &  &  & (0.02) & (0.03) & (0.03) & (0.01) & (0.22) & (0.23) & (4.66) & (4.66) \\ \cline{3-11} 
 &  & \multirow{2}{*}{1000} & 12.06 & 12.31 & 13.05 & 2.22 & 12.20 & 11.80 & 1698.00 & 1698.00 \\
 &  &  & (0.13) & (0.13) & (0.15) & (0.02) & (0.13) & (0.13) & (5.80) & (5.80) \\ \cline{3-11} 
 &  & \multirow{2}{*}{2500} & 78.69 & 203.18 & 225.83 & 25.51 & 12.50 & 12.10 & 4268.10 & 4268.10 \\
 &  &  & (1.56) & (3.53) & (4.36) & (0.48) & (0.22) & (0.23) & (11.70) & (11.70) \\ \cline{3-11} 
 &  & \multirow{2}{*}{5000} & 350.79 & 1718.20 & 1901.72 & 182.99 & 13.00 & 12.30 & 8521.50 & 8521.50 \\
 &  &  & (7.29) & (35.40) & (39.33) & (3.18) & (0.30) & (0.21) & (11.65) & (11.65) \\ \cline{2-11} 
 & \multirow{8}{*}{2000} & \multirow{2}{*}{500} & 1.12 & 1.64 & 1.62 & 0.45 & 11.80 & 11.00 & 854.20 & 854.20 \\
 &  &  & (0.01) & (0.02) & (0.00) & (0.01) & (0.13) & (0.00) & (3.14) & (3.14) \\ \cline{3-11} 
 &  & \multirow{2}{*}{1000} & 10.81 & 12.22 & 12.76 & 2.11 & 11.60 & 11.10 & 1717.00 & 1717.00 \\
 &  &  & (0.21) & (0.16) & (0.11) & (0.02) & (0.16) & (0.10) & (4.96) & (4.96) \\ \cline{3-11} 
 &  & \multirow{2}{*}{2500} & 159.15 & 204.24 & 216.04 & 24.00 & 12.00 & 11.30 & 4291.10 & 4291.10 \\
 &  &  & (0.09) & (0.27) & (1.86) & (0.32) & (0.00) & (0.15) & (5.94) & (5.94) \\ \cline{3-11} 
 &  & \multirow{2}{*}{5000} & 637.67 & 1597.47 & 1796.61 & 171.11 & 12.10 & 11.50 & 8578.20 & 8578.30 \\
 &  &  & (5.89) & (15.28) & (32.73) & (2.46) & (0.10) & (0.17) & (14.48) & (14.51) \\ \hline
\end{tabular}
}
\par\end{centering}
\end{table}

\begin{table}[!htb]
\begin{centering} 
\caption{Average
computation time (in seconds),  number of iterations, 
and number of estimated edges for the scale-free network.
Numbers within parentheses denote standard errors.} \label{tb:sf} \medskip 
\scalebox{0.75}{
\begin{tabular}{ccccccccccc}
\hline
\multirow{2}{*}{$\lambda$} & \multirow{2}{*}{$n$} & \multirow{2}{*}{$p$} & \multicolumn{4}{c}{Computation time (sec.)} & \multicolumn{2}{c}{Iteration} & \multicolumn{2}{c}{$|\hat{E}|$} \\ \cline{4-11} 
 &  &  & CD-BLAS & CD-NAIVE & PCD-CPU & PCD-GPU & CD & PCD & CD & PCD \\ \hline
\multirow{24}{*}{0.1} & \multirow{8}{*}{500} & \multirow{2}{*}{500} & 1.33 & 1.35 & 1.53 & 0.46 & 11.30 & 12.20 & 2348.30 & 2348.30 \\
 &  &  & (0.02) & (0.02) & (0.04) & (0.01) & (0.15) & (0.29) & (13.88) & (13.88) \\ \cline{3-11} 
 &  & \multirow{2}{*}{1000} & 5.64 & 11.41 & 13.18 & 2.34 & 11.90 & 12.60 & 8090.90 & 8090.90 \\
 &  &  & (0.15) & (0.30) & (0.28) & (0.05) & (0.31) & (0.27) & (22.41) & (22.41) \\ \cline{3-11} 
 &  & \multirow{2}{*}{2500} & 43.59 & 228.30 & 269.74 & 30.84 & 14.20 & 14.60 & 42601.50 & 42601.50 \\
 &  &  & (1.59) & (8.18) & (9.14) & (1.05) & (0.51) & (0.50) & (34.34) & (34.34) \\ \cline{3-11} 
 &  & \multirow{2}{*}{5000} & 193.67 & 1872.72 & 2345.42 & 230.10 & 14.10 & 15.50 & 144508.00 & 144507.70 \\
 &  &  & (3.84) & (37.21) & (39.78) & (3.98) & (0.28) & (0.27) & (58.82) & (58.88) \\ \cline{2-11} 
 & \multirow{8}{*}{1000} & \multirow{2}{*}{500} & 0.96 & 1.44 & 1.61 & 0.48 & 11.60 & 12.50 & 598.20 & 598.20 \\
 &  &  & (0.02) & (0.03) & (0.03) & (0.01) & (0.27) & (0.27) & (5.46) & (5.46) \\ \cline{3-11} 
 &  & \multirow{2}{*}{1000} & 10.76 & 11.00 & 12.49 & 2.19 & 11.20 & 11.70 & 1340.00 & 1340.00 \\
 &  &  & (0.28) & (0.27) & (0.22) & (0.04) & (0.29) & (0.21) & (4.94) & (4.94) \\ \cline{3-11} 
 &  & \multirow{2}{*}{2500} & 92.85 & 235.54 & 269.49 & 29.01 & 13.90 & 13.70 & 4471.70 & 4471.70 \\
 &  &  & (2.22) & (5.80) & (3.32) & (0.32) & (0.31) & (0.15) & (18.16) & (18.16) \\ \cline{3-11} 
 &  & \multirow{2}{*}{5000} & 369.05 & 1829.25 & 2121.88 & 209.48 & 13.80 & 14.10 & 12557.50 & 12557.60 \\
 &  &  & (7.76) & (38.16) & (58.87) & (6.02) & (0.29) & (0.41) & (21.90) & (21.94) \\ \cline{2-11} 
 & \multirow{8}{*}{2000} & \multirow{2}{*}{500} & 1.09 & 1.60 & 1.77 & 0.49 & 11.90 & 12.80 & 506.50 & 506.50 \\
 &  &  & (0.01) & (0.02) & (0.02) & (0.01) & (0.18) & (0.20) & (0.50) & (0.50) \\ \cline{3-11} 
 &  & \multirow{2}{*}{1000} & 10.62 & 11.97 & 12.90 & 2.18 & 11.70 & 11.60 & 1011.00 & 1011.00 \\
 &  &  & (0.25) & (0.20) & (0.32) & (0.05) & (0.21) & (0.31) & (1.02) & (1.02) \\ \cline{3-11} 
 &  & \multirow{2}{*}{2500} & 187.80 & 239.21 & 257.65 & 28.62 & 14.50 & 13.50 & 2517.50 & 2517.50 \\
 &  &  & (5.47) & (6.88) & (4.20) & (0.47) & (0.43) & (0.22) & (1.56) & (1.56) \\ \cline{3-11} 
 &  & \multirow{2}{*}{5000} & 680.01 & 1705.53 & 2078.65 & 209.39 & 13.10 & 14.10 & 5045.90 & 5045.90 \\
 &  &  & (9.11) & (23.24) & (26.84) & (2.66) & (0.18) & (0.18) & (2.74) & (2.74) \\ \hline
\multirow{24}{*}{0.3} & \multirow{8}{*}{500} & \multirow{2}{*}{500} & 1.03 & 1.06 & 1.14 & 0.37 & 8.80 & 9.00 & 364.70 & 364.70 \\
 &  &  & (0.02) & (0.02) & (0.02) & (0.00) & (0.13) & (0.15) & (1.69) & (1.69) \\ \cline{3-11} 
 &  & \multirow{2}{*}{1000} & 4.43 & 9.06 & 10.14 & 1.80 & 9.40 & 9.60 & 713.30 & 713.30 \\
 &  &  & (0.08) & (0.16) & (0.22) & (0.04) & (0.16) & (0.22) & (2.13) & (2.13) \\ \cline{3-11} 
 &  & \multirow{2}{*}{2500} & 34.08 & 176.28 & 208.13 & 22.30 & 10.50 & 10.50 & 1755.40 & 1755.40 \\
 &  &  & (0.78) & (3.97) & (3.55) & (0.46) & (0.27) & (0.22) & (5.31) & (5.31) \\ \cline{3-11} 
 &  & \multirow{2}{*}{5000} & 138.34 & 1343.89 & 1568.08 & 157.45 & 10.40 & 10.60 & 3569.40 & 3569.40 \\
 &  &  & (2.87) & (27.90) & (43.39) & (4.52) & (0.22) & (0.31) & (6.12) & (6.12) \\ \cline{2-11} 
 & \multirow{8}{*}{1000} & \multirow{2}{*}{500} & 0.77 & 1.15 & 1.23 & 0.37 & 9.00 & 9.30 & 367.80 & 367.80 \\
 &  &  & (0.00) & (0.00) & (0.02) & (0.00) & (0.00) & (0.15) & (1.50) & (1.50) \\ \cline{3-11} 
 &  & \multirow{2}{*}{1000} & 8.65 & 8.94 & 9.75 & 1.71 & 9.00 & 9.00 & 715.00 & 715.00 \\
 &  &  & (0.15) & (0.14) & (0.16) & (0.03) & (0.15) & (0.15) & (2.09) & (2.09) \\ \cline{3-11} 
 &  & \multirow{2}{*}{2500} & 68.93 & 176.52 & 198.19 & 21.84 & 10.50 & 10.30 & 1758.80 & 1758.80 \\
 &  &  & (1.09) & (2.82) & (3.00) & (0.32) & (0.17) & (0.15) & (2.63) & (2.63) \\ \cline{3-11} 
 &  & \multirow{2}{*}{5000} & 267.18 & 1323.10 & 1530.87 & 150.17 & 9.90 & 10.10 & 3582.60 & 3582.60 \\
 &  &  & (2.69) & (13.37) & (15.77) & (1.48) & (0.10) & (0.10) & (3.96) & (3.96) \\ \cline{2-11} 
 & \multirow{8}{*}{2000} & \multirow{2}{*}{500} & 0.87 & 1.26 & 1.33 & 0.38 & 9.00 & 9.10 & 367.30 & 367.30 \\
 &  &  & (0.00) & (0.00) & (0.01) & (0.00) & (0.00) & (0.10) & (1.24) & (1.24) \\ \cline{3-11} 
 &  & \multirow{2}{*}{1000} & 8.21 & 9.53 & 10.43 & 1.75 & 9.10 & 9.20 & 712.70 & 712.70 \\
 &  &  & (0.07) & (0.09) & (0.14) & (0.02) & (0.10) & (0.13) & (1.73) & (1.73) \\ \cline{3-11} 
 &  & \multirow{2}{*}{2500} & 134.24 & 173.16 & 199.99 & 22.10 & 10.40 & 10.40 & 1760.00 & 1760.00 \\
 &  &  & (2.07) & (2.64) & (3.10) & (0.34) & (0.16) & (0.16) & (3.50) & (3.50) \\ \cline{3-11} 
 &  & \multirow{2}{*}{5000} & 513.17 & 1294.45 & 1482.85 & 148.71 & 9.90 & 10.00 & 3585.10 & 3585.10 \\
 &  &  & (5.14) & (13.08) & (1.64) & (0.00) & (0.10) & (0.00) & (4.13) & (4.13) \\ \hline
\end{tabular}
}
\par\end{centering}
\end{table}

\begin{table}[!htb]
\caption{Average
computation time in seconds (Comp. Time),  
and number of estimated edges ($|\hat{E}|$) for the AR(2) and scale-free networks
for PCD-CPU, PCD-GPU, GLASSO and FASTCLIME.
Numbers within parentheses denote standard errors.
} \label{tb:add} \medskip   
\begin{centering}
\scalebox{0.8}{
\begin{tabular}{|c|c|c|c|c|c|c|c|c|}
\hline
\multirow{2}{*}{Network} & \multirow{2}{*}{$p$} & \multirow{2}{*}{Model} & \multicolumn{3}{c|}{$n=500$} & \multicolumn{3}{c|}{$n=1000$} \\ \cline{4-9}
 &  &  &  $\lambda$ &  $|\hat{E}|$ & Comp. Time & $\lambda$ &  $|\hat{E}|$ & Comp. Time  \\ \hline
\multirow{16}{*}{AR(2)} & \multirow{8}{*}{$500$} & \multirow{2}{*}{PCD-CPU} & \multirow{2}{*}{0.3} & 859.5 & 1.73 & \multirow{2}{*}{0.3} & 1722.2 & 14.80 \\
&  &  &  & (5.00) & (0.06) &  & (5.87) & (0.45) \\ \cline{3-9}
&  & \multirow{2}{*}{PCD-GPU} & \multirow{2}{*}{0.3} & 859.5 & 0.50 & \multirow{2}{*}{0.3} & 1722.2 & 2.55 \\
&  &  &  & (5.00) & (0.02) &  & (5.87) & (0.07) \\ \cline{3-9}
&  & \multirow{2}{*}{GLASSO} & \multirow{2}{*}{0.383} & 860.4 & 0.97 & \multirow{2}{*}{0.386} & 1711.4 & 7.57  \\
&  &  &  & (3.50) & (0.00) &  & (6.36) & (0.02) \\ \cline{3-9}
&  & \multirow{2}{*}{FASTCLIME} &  \multirow{2}{*}{0.311} & 867.3 & 26.91  & \multirow{2}{*}{0.312} & 1701.2 & 198.88 \\
&  &  &  & (3.48) & (0.16) &  & (7.04) & (0.61) \\ \cline{2-9}

& \multirow{8}{*}{$1000$} & \multirow{2}{*}{PCD-CPU} & \multirow{2}{*}{0.3} & 853.6 & 1.62 & \multirow{2}{*}{0.3} & 1698.0 & 13.05 \\
&  &  &  & (4.66) & (0.03) &  & (5.80) & (0.15) \\ \cline{3-9}
&  & \multirow{2}{*}{PCD-GPU} & \multirow{2}{*}{0.3} & 853.6 & 0.46 & \multirow{2}{*}{0.3} & 1698.0 & 2.22 \\
&  &  &  & (4.66) & (0.01) &  & (5.80) & (0.02) \\ \cline{3-9}
&  & \multirow{2}{*}{GLASSO} & \multirow{2}{*}{0.384} & 861.5 & 1.03 & \multirow{2}{*}{0.388} & 1699.6 & 7.85 \\
&  &  &  & (3.54) & (0.00) &  & (5.16) & (0.03) \\ \cline{3-9}
&  & \multirow{2}{*}{FASTCLIME} &  \multirow{2}{*}{0.311} & 854.0 & 27.25 & \multirow{2}{*}{0.315} & 1698.4 & 199.42 \\
&  &  &  & (5.06) & (0.11) &  & (5.35) & (0.51) \\ \hline

\multirow{16}{*}{Scale-free} & \multirow{8}{*}{$500$} & \multirow{2}{*}{PCD-CPU} & \multirow{2}{*}{0.3} & 364.7 & 1.14 & \multirow{2}{*}{0.3} & 713.3 & 10.14 \\
&  &  &  & (1.69) & (0.02) &  & (2.13) & (0.22) \\ \cline{3-9}
&  & \multirow{2}{*}{PCD-GPU} & \multirow{2}{*}{0.3} & 364.7 & 0.37 & \multirow{2}{*}{0.3} & 713.3 & 1.80 \\
&  &  &  & (1.69) & (0.00) &  & (2.13) & (0.04) \\ \cline{3-9}
&  & \multirow{2}{*}{GLASSO} & \multirow{2}{*}{0.241} & 365.5 & 0.13 & \multirow{2}{*}{0.249} & 719.6 & 0.79 \\
&  &  &  & (2.23) & (0.00) &  & (1.86) & (0.00) \\ \cline{3-9}
&  & \multirow{2}{*}{FASTCLIME} &  \multirow{2}{*}{0.236} & 364.7 & 27.29 & \multirow{2}{*}{0.237} & 717.8 & 192.89 \\
&  &  &  & (1.56) & (0.07) &  & (1.81) & (0.18) \\ \cline{2-9}

& \multirow{8}{*}{$1000$} & \multirow{2}{*}{PCD-CPU} & \multirow{2}{*}{0.3} & 367.8 & 1.23 & \multirow{2}{*}{0.3} & 715.0 & 9.75 \\
&  &  &  & (1.50) & (0.02) &  & (2.09) & (0.16) \\ \cline{3-9}
&  & \multirow{2}{*}{PCD-GPU} & \multirow{2}{*}{0.3} & 367.8 & 0.37 & \multirow{2}{*}{0.3} & 715.0 & 1.71 \\
&  &  &  & (1.50) & (0.00) &  & (2.09) & (0.03) \\ \cline{3-9}
&  & \multirow{2}{*}{GLASSO} & \multirow{2}{*}{0.244} & 368.5 & 0.19 & \multirow{2}{*}{0.249} & 707.6 & 1.01 \\
&  &  &  & (1.67) & (0.00) &  & (2.50) & (0.00) \\ \cline{3-9}
&  & \multirow{2}{*}{FASTCLIME} &  \multirow{2}{*}{0.235} & 366.6 & 27.48 & \multirow{2}{*}{0.237} & 712.6 & 184.79 \\
&  &  &  & (1.71) & (0.11) &  & (2.08) & (1.32) \\ \hline

\end{tabular}%
}
\par\end{centering}
\end{table}


\begin{figure}[!h]
\begin{center}
  \begin{minipage}[b]{.45\linewidth}
  \centering   \centerline{\includegraphics[scale=0.4]{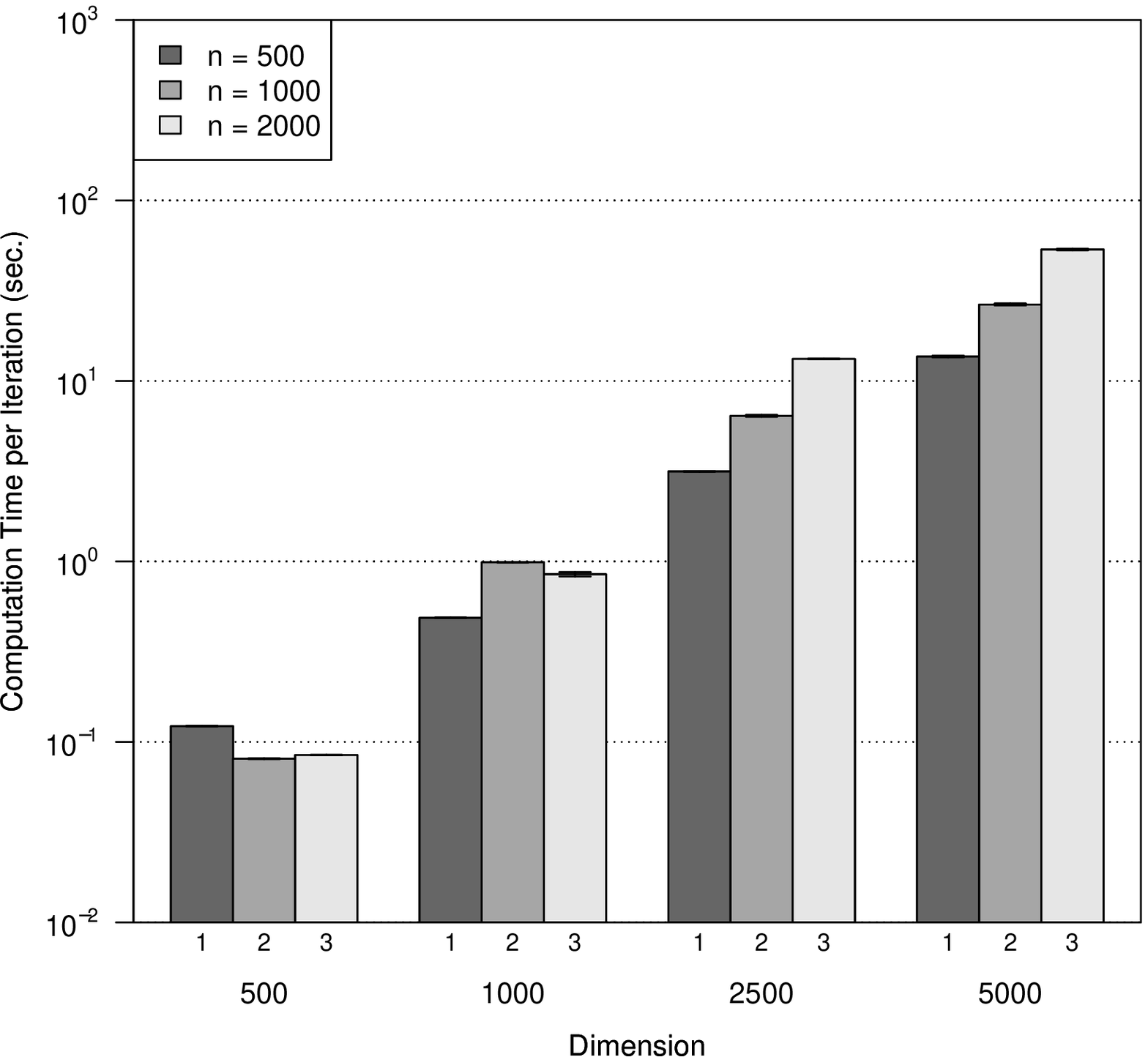}}
  \centerline{(a) CD-BLAS, $\lambda=0.1$ }
\end{minipage}
 \begin{minipage}[b]{.45\linewidth}
  \centering   \centerline{\includegraphics[scale=0.4]{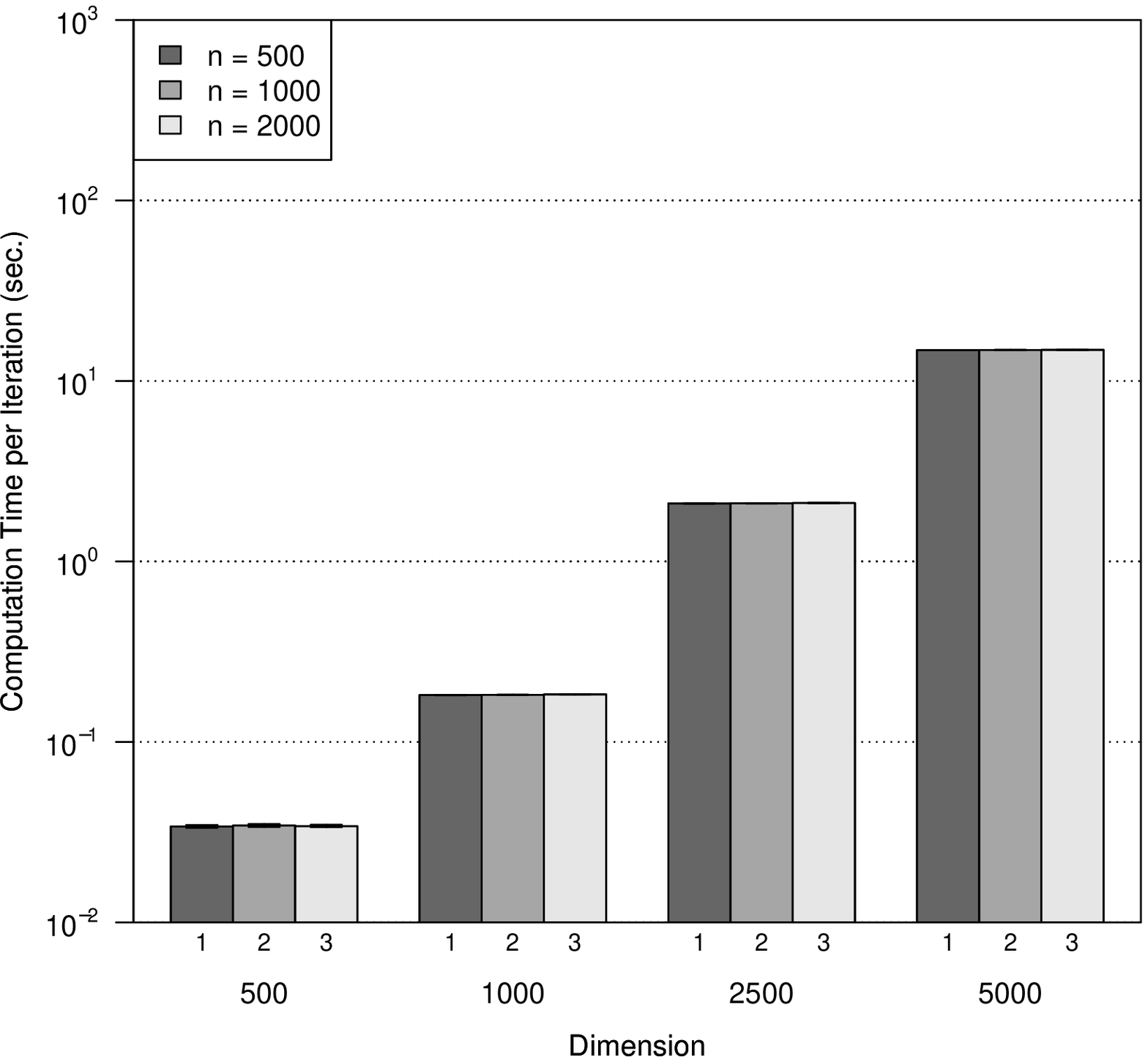}}
  \centerline{(b) PCD-GPU, $\lambda=0.1$ }
 \end{minipage}.
\begin{minipage}[b]{.45\linewidth}
  \centering   \centerline{\includegraphics[scale=0.4]{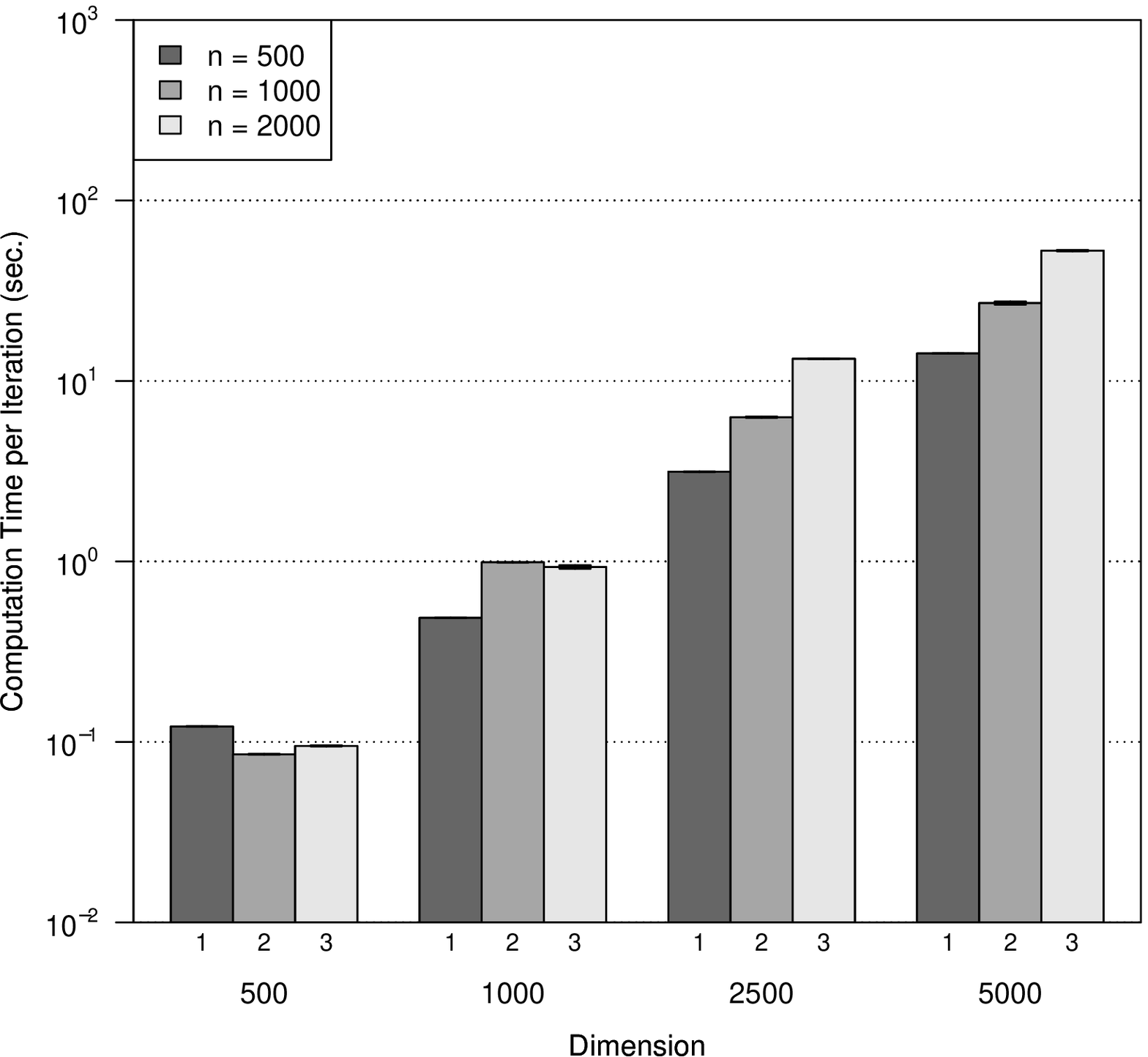}}
   \centerline{(c) CD-BLAS, $\lambda=0.3$ }
 \end{minipage}
\begin{minipage}[b]{.45\linewidth}
  \centering   \centerline{\includegraphics[scale=0.4]{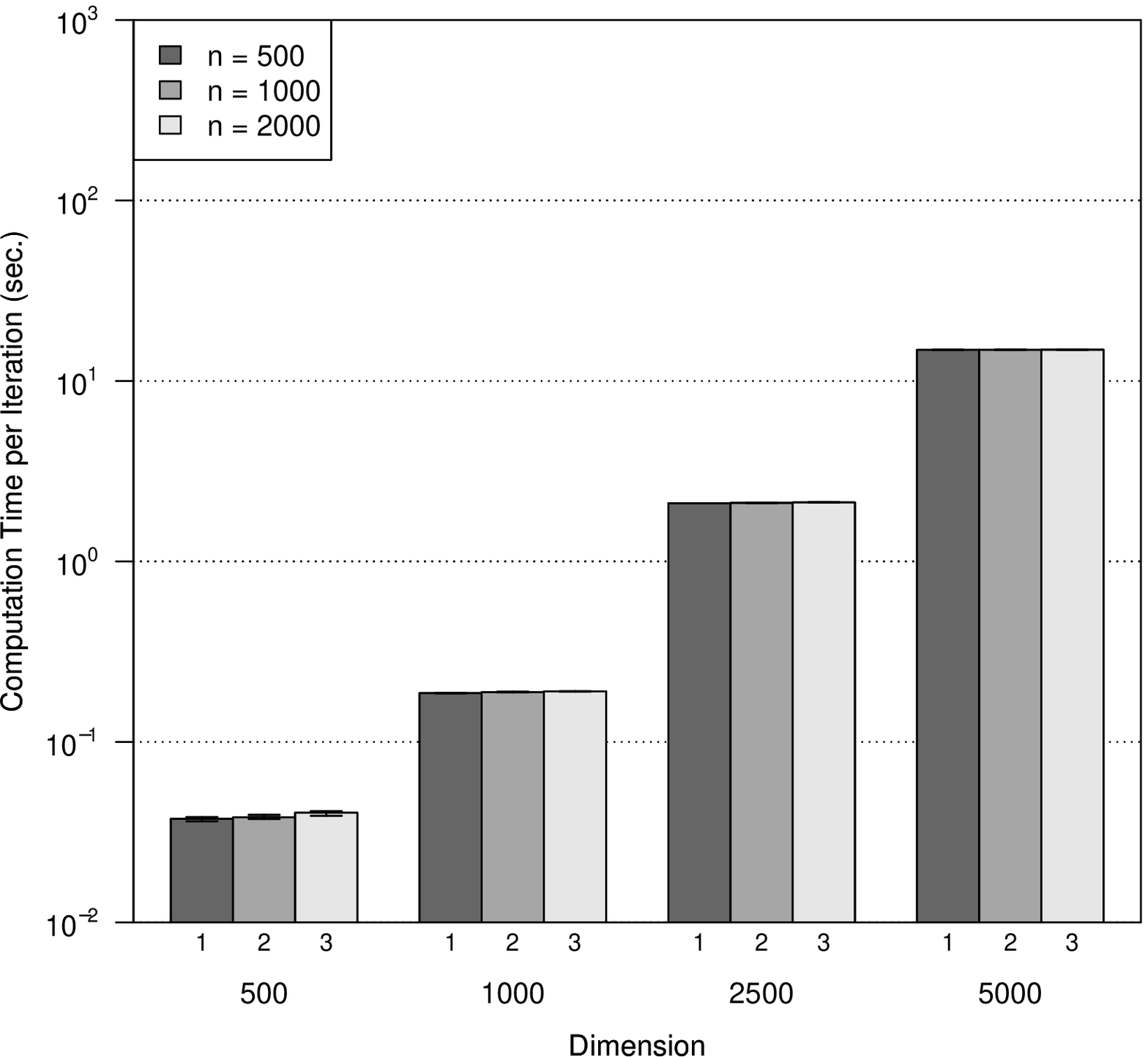}}
  \centerline{(d) PCD-GPU, $\lambda=0.3$ }
\end{minipage}
\caption{Average computation time per iteration for the AR(2) network. The vertical lines denote 95$\%$ confidence intervals of the mean computation time per iteration.} \label{fig:ar2}
\end{center}
\end{figure}

\begin{figure}[!h]
\begin{center}
  \begin{minipage}[b]{.45\linewidth}
  \centering   \centerline{\includegraphics[scale=0.4]{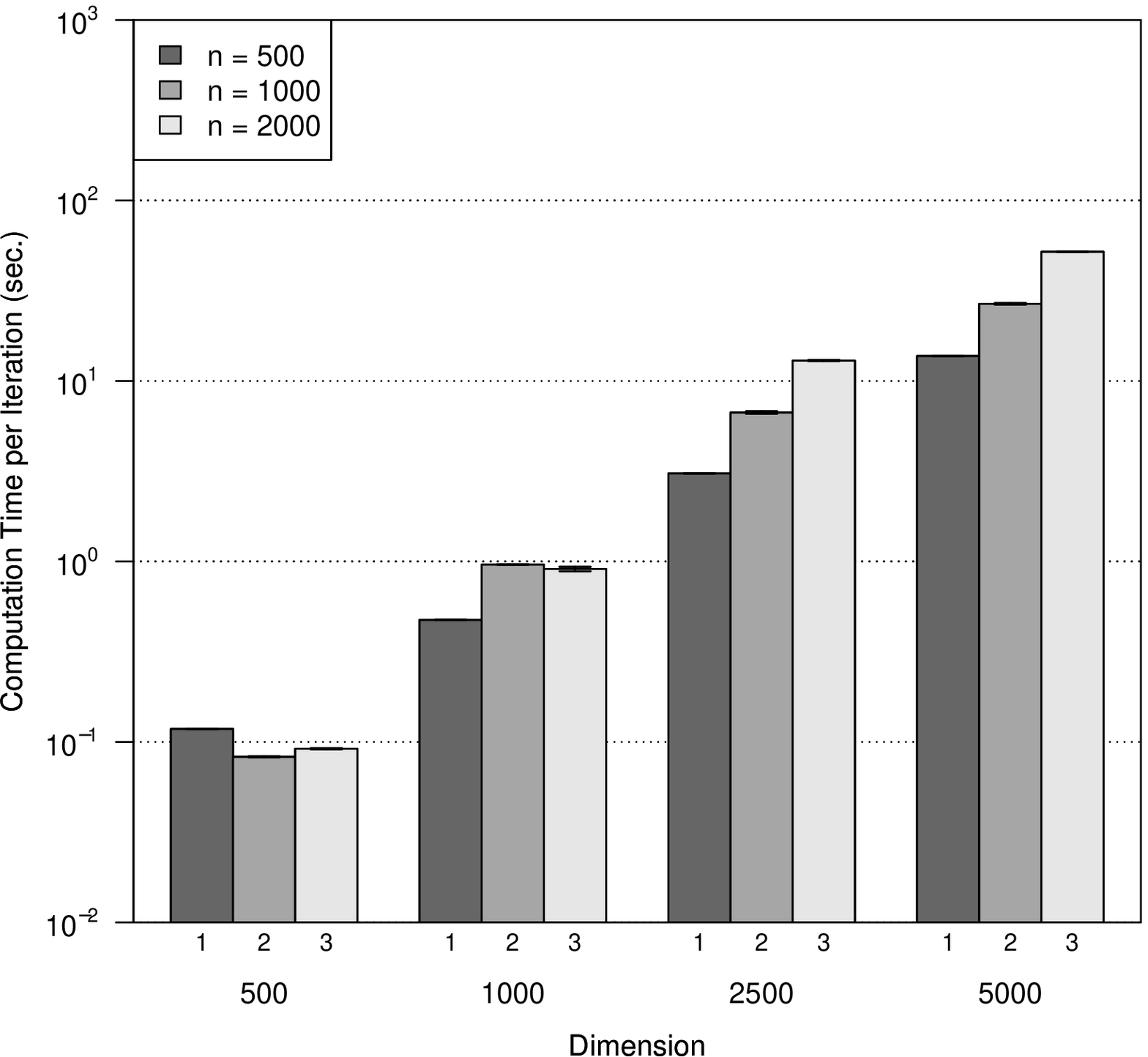}}
  \centerline{(a) CD-BLAS, $\lambda=0.1$ }
\end{minipage}
 \begin{minipage}[b]{.45\linewidth}
  \centering   \centerline{\includegraphics[scale=0.4]{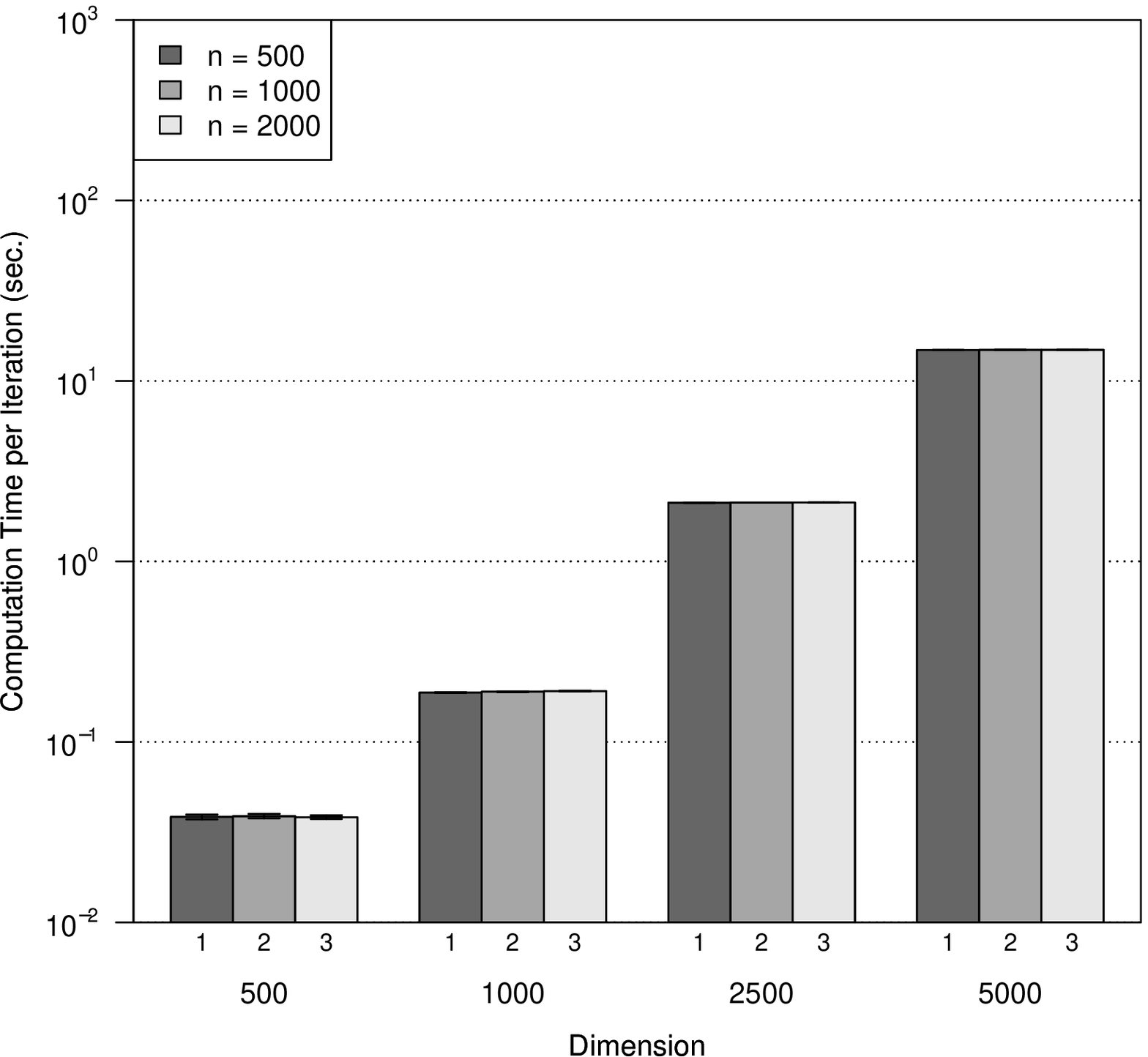}}
  \centerline{(b) PCD-GPU, $\lambda=0.1$ }
 \end{minipage}.
\begin{minipage}[b]{.45\linewidth}
  \centering   \centerline{\includegraphics[scale=0.4]{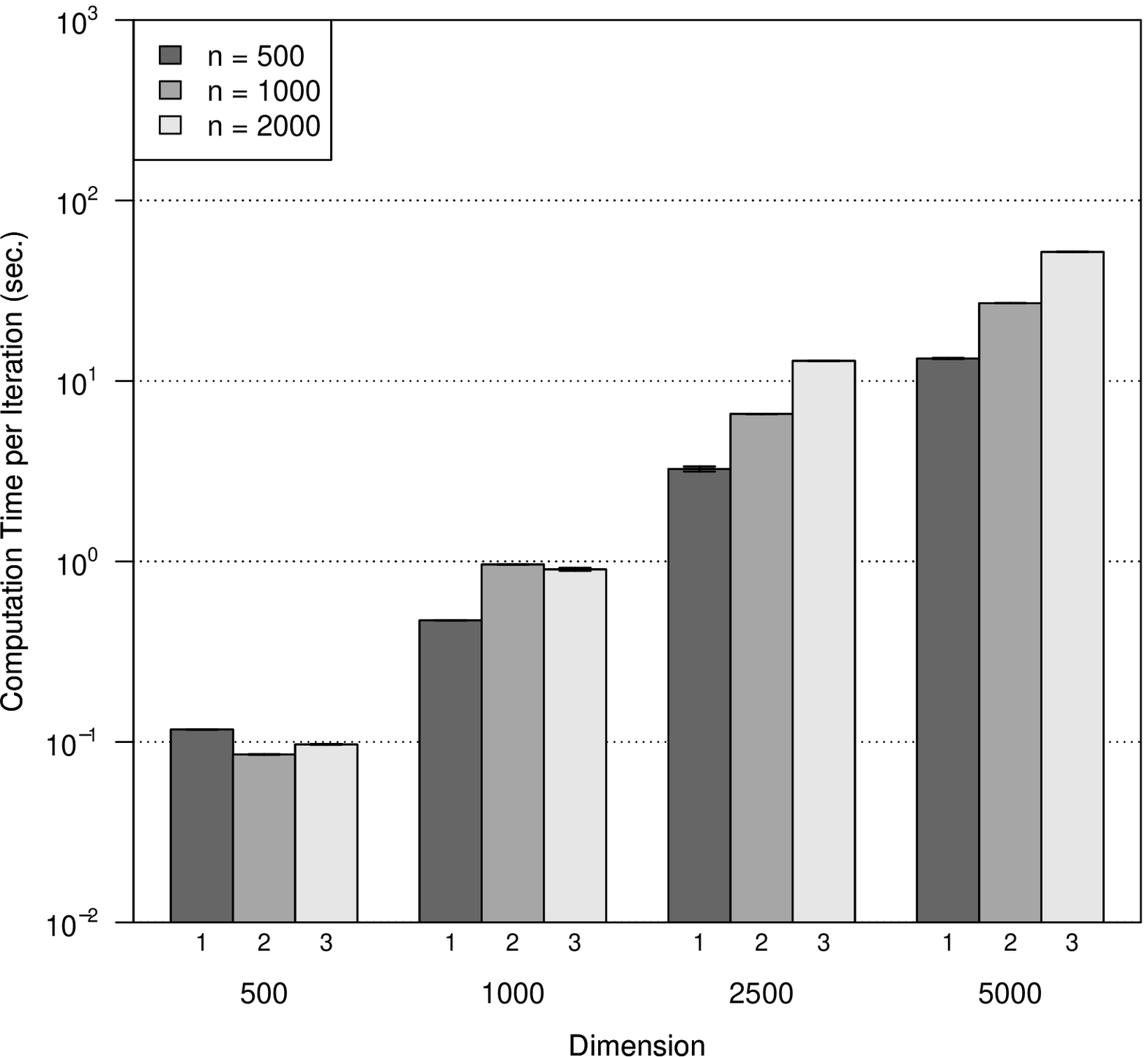}}
   \centerline{(c) CD-BLAS, $\lambda=0.3$ }
 \end{minipage}
\begin{minipage}[b]{.45\linewidth}
  \centering   \centerline{\includegraphics[scale=0.4]{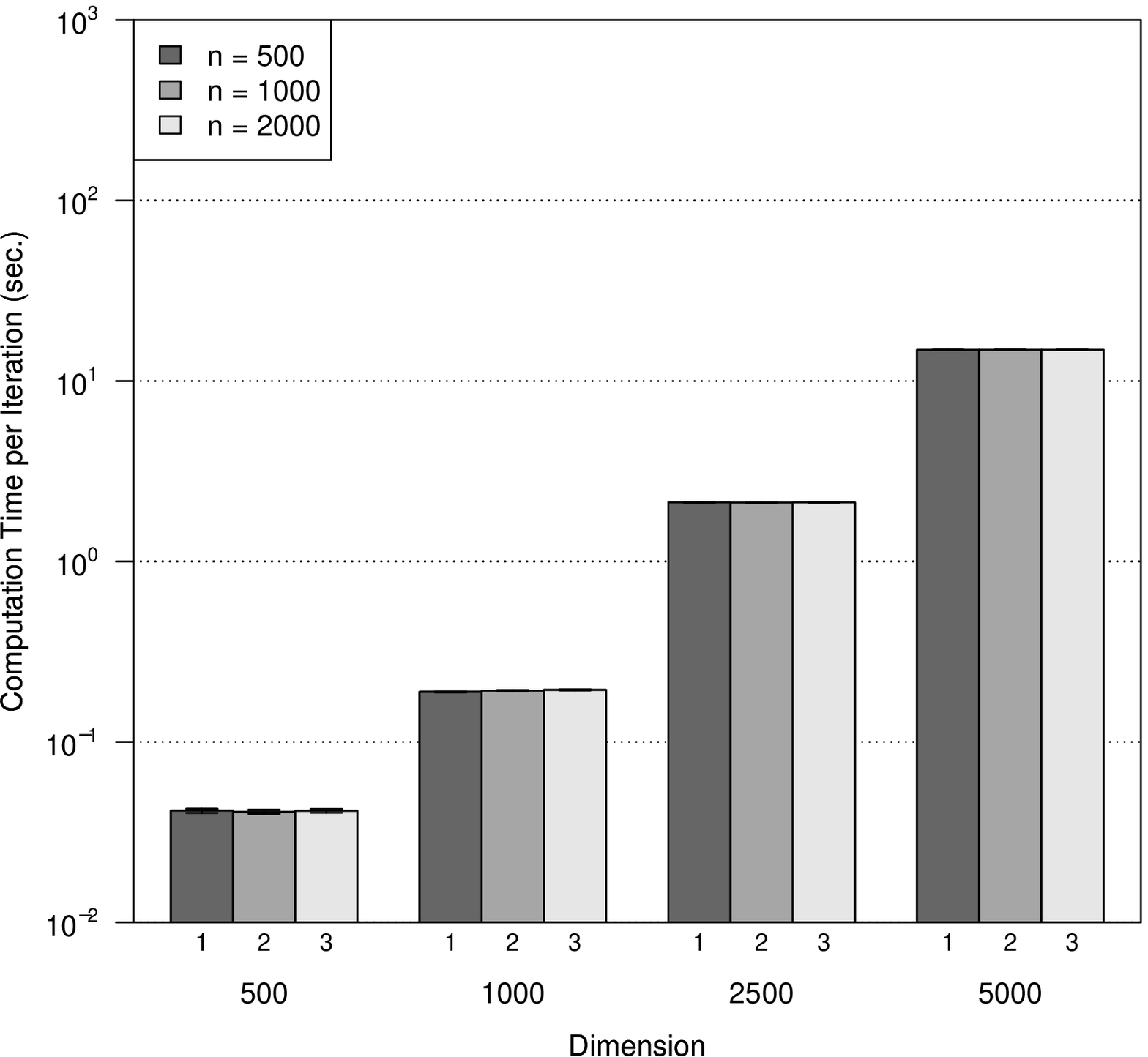}}
  \centerline{(d) PCD-GPU, $\lambda=0.3$ }
\end{minipage}
\caption{Average computation time per iteration for the scale-free network. The vertical lines denote 95$\%$ confidence intervals of the mean computation time per iteration.} \label{fig:sf}
\end{center}
\end{figure}

\section{Concluding remarks}\label{sec:final}

In this paper, we proposed the parallel coordinate
descent algorithm for CONCORD, which simultaneously updates $p_{even}/2$ elements,
which is $p/2$ for an even $p$ and
$(p-1)/2$ for an odd $p$.
We also showed, by applying the theoretical results to edge coloring,  that $p_{even}/2$
is the maximum number of simultaneously
updatable off-diagonal elements in the CONCORD-CD algorithm.
Comprehensive numerical studies
show that the proposed CONCORD-PCD algorithm
is adequate for GPU-parallel computation,
and more efficient than the original CONCORD-CD algorithm,
for large datasets.

We conclude the paper with discussion about possible extensions.
Our idea of parallelized coordinate descent 
can be applied to modeling gene regulatory networks from 
heterogeneous data through joint estimation of sparse precision matrices
\citep{Danaher2014}.
For example, let us consider
the following objective function, which estimates two precision matrices, $\Omega_1=(\omega_{ij}^{(1)})$
and $\Omega_2=(\omega_{ij}^{(2)})$,
under the constraint that
both matrices are sparse and 
only slightly different from each other:
\begin{equation} \nonumber
\begin{array}{rcl} \displaystyle
L_{joint}(\Omega_1, \Omega_2; \lambda_1, \lambda_2) &=& \displaystyle
 \sum_{m=1}^2 \Big\{-\sum_{i=1}^p n\log \omega_{ii}^{(m)} + \frac{1}{2}\sum_{i=1}^p \sum_{k=1}^n \Big( \omega_{ii}^{(m)} X_{ki}^m + \sum_{j\neq i} \omega_{ij}^{(m)}  X_{kj}^m \Big)^2 \Big\} \\
 && \displaystyle
 + \lambda_1 \sum_{m=1}^2
\sum_{i<j} |\omega_{ij}^{(m)}|
+\lambda_2 \sum_{i \leq j} |\omega_{ij}^{(1)}
- \omega_{ij}^{(2)}|,
\end{array}
\end{equation}
where $X_{ki}^m$ is the $(k,i)$th element
of the observed dataset from $m$th population ($m=1,2$).
Consider a block coordinate descent algorithm that minimizes along  $(\omega_{ij}^{(1)}, \omega_{ij}^{(2)})$
for each update, in which the update formula has a closed-form expression 
similar to one in  \cite{Yu2018}.
One can show that if two edge indices $ij$ and $i'j'$ are disjoint,
then the update formula for $(\hat{\omega}_{ij}^{(1)}, \hat{\omega}_{ij}^{(2)})$ 
does not involve $(\hat{\omega}_{i'j'}^{(1)}, \hat{\omega}_{i'j'}^{(2)})$. Thus,
one can develop a parallelization for this algorithm as presented in this paper.

\begin{acknowledgements}

This research was supported by the National Research Foundation of Korea (NRF-2018R1C1B6001108), Inha University Research Grant, and Sookmyung Women's University Research Grant (No. 1-2003-2004).

\end{acknowledgements}

%
%

\clearpage

\bibliographystyle{spbasic}      
\bibliography{ref.bib}

\begin{thebibliography}{}

\bibitem[Barab{\'{a}}si and Albert, 1999]{Barabasi1999}
Barab{\'{a}}si, A.-L. and Albert, R. (1999).
\newblock {Emergence of scaling in random networks}.
\newblock {\em science}, 286(5439):509--512.

\bibitem[Bradley et~al., 2011]{Bradley2011}
Bradley, J.~K., Kyrola, A., Bickson, D., and Guestrin, C. (2011).
\newblock {Parallel coordinate descent for L1-regularized loss minimization}.
\newblock {\em Proceedings of the 28th International Conference on Machine
  Learning, ICML 2011}, (1998):321--328.

\bibitem[Cai et~al., 2011]{Cai2011}
Cai, T., Liu, W., and Luo, X. (2011).
\newblock {A constrained l1 minimization approach to sparse precision matrix
  estimation}.
\newblock {\em Journal of the American Statistical Association},
  106(494):594--607.

\bibitem[Cai et~al., 2016]{Cai2016b}
Cai, T.~T., Liu, W., and Zhou, H.~H. (2016).
\newblock {Estimating sparse precision matrix: Optimal rates of convergence and
  adaptive estimation}.
\newblock {\em The Annals of Statistics}, 44(2):455--488.

\bibitem[Danaher et~al., 2014]{Danaher2014}
Danaher, P., Wang, P., and Witten, D.~M. (2014).
\newblock {The joint graphical lasso for inverse covariance estimation across
  multiple classes.}
\newblock {\em Journal of the Royal Statistical Society. Series B, Statistical
  methodology}, 76(2):373--397.

\bibitem[Dinitz et~al., 2006]{Dinitz2006}
Dinitz, J.~H., Froncek, D., Lamken, E.~R., and Wallis, W.~D. (2006).
\newblock {Scheduling a tournament}.
\newblock In {\em Handbook of Combinatorial Designs}, chapter VI.51, pages
  591--606. Chapman {\&} Hall/CRC, second ed. edition.

\bibitem[Formanowicz and Tana{\'{s}}, 2012]{Formanowicz2012}
Formanowicz, P. and Tana{\'{s}}, K. (2012).
\newblock {A survey of graph coloring - its types, methods and applications}.
\newblock {\em Foundations of Computing and Decision Sciences}, 37(3):223--238.

\bibitem[Friedman et~al., 2008]{Friedman2008}
Friedman, J., Hastie, T., and Tibshirani, R. (2008).
\newblock {Sparse inverse covariance estimation with the graphical lasso}.
\newblock {\em Biostatistics}, 9(3):432--441.

\bibitem[Hsieh, 2014]{Hsieh2014}
Hsieh, C.-j. (2014).
\newblock {QUIC : Quadratic Approximation for Sparse Inverse Covariance
  Estimation}.
\newblock {\em Journal of Machine Learning Research}, 15:2911--2947.

\bibitem[Hsieh et~al., 2013]{Hsieh2013}
Hsieh, C.-J., Sustik, M.~A., Dhillon, I.~S., Ravikumar, P.~K., and Poldrack, R.
  (2013).
\newblock {BIG {\&} QUIC: Sparse Inverse Covariance Estimation for a Million
  Variables}.
\newblock In Burges, C. J.~C., Bottou, L., Welling, M., Ghahramani, Z., and
  Weinberger, K.~Q., editors, {\em Advances in Neural Information Processing
  Systems 26}, pages 3165--3173. Curran Associates, Inc.

\bibitem[Khare et~al., 2015]{Khare2015}
Khare, K., Oh, S.-Y., and Rajaratnam, B. (2015).
\newblock {A convex pseudolikelihood framework for high dimensional partial
  correlation estimation with convergence guarantees}.
\newblock {\em Journal of the Royal Statistical Society: Series B (Statistical
  Methodology)}, 77(4):803--825.

\bibitem[Lawson et~al., 1979]{Lawson1979}
Lawson, C., Hanson, R., Kincaid, D., and Krogh, F. (1979).
\newblock {Algorithm 539: Basic linear algebra subprograms for Fortran usage}.
\newblock {\em ACM Transactions on Mathematical Software}, 5(3):308--323.

\bibitem[Mazumder and Hastie, 2012]{Mazumder2012}
Mazumder, R. and Hastie, T. (2012).
\newblock {The graphical lasso: New insights and alternatives}.
\newblock {\em Electronic Journal of Statistics}, 6(August):2125--2149.

\bibitem[Meinshausen and B{\"{u}}hlmann, 2006]{Meinshausen2006}
Meinshausen, N. and B{\"{u}}hlmann, P. (2006).
\newblock {High-dimensional graphs and variable selection with the Lasso}.
\newblock {\em Annals of Statistics}, 34(3):1436--1462.

\bibitem[Nakano et~al., 1995]{Nakano1995}
Nakano, S.-i., Zhou, X., and Nishizeki, T. (1995).
\newblock {Edge-coloring algorithms}.
\newblock In {\em Computer Science Today. Lecture Notes in Computer Science,
  vol. 1000}, pages 172--183. Springer, Berlin, Heidelberg.

\bibitem[Newman, 2003]{Newman2003}
Newman, M. E.~J. (2003).
\newblock {The structure and function of complex networks}.
\newblock {\em SIAM review}, 45(2):167--256.

\bibitem[Pang et~al., 2014]{Pang2014}
Pang, H., Liu, H., and Vanderbei, R. (2014).
\newblock The fastclime package for linear programming and large-scale
  precision matrix estimation in r.
\newblock {\em Journal of Machine Learning Research}, 15:489--493.

\bibitem[Peng et~al., 2009]{Peng2009}
Peng, J., Wang, P., Zhou, N., and Zhu, J. (2009).
\newblock {Partial correlation estimation by joint sparse regression models}.
\newblock {\em Journal of the American Statistical Association},
  104(486):735--746.

\bibitem[Richt{\'{a}}rik and Tak{\'{a}}{\v{c}}, 2016]{Richtarik2016a}
Richt{\'{a}}rik, P. and Tak{\'{a}}{\v{c}}, M. (2016).
\newblock {\em {Parallel coordinate descent methods for big data
  optimization}}, volume 156.

\bibitem[Sun and Zhang, 2013]{Sun2013}
Sun, T. and Zhang, C.~H. (2013).
\newblock {Sparse matrix inversion with scaled lasso}.
\newblock {\em Journal of Machine Learning Research}, 14:3385--3418.

\bibitem[Tseng, 2001]{Tseng2001}
Tseng, P. (2001).
\newblock {Convergence of a block coordinate descent method for
  nondifferentiable minimization}.
\newblock {\em Journal of Optimization Theory and Applications},
  109(3):475--494.

\bibitem[Wang et~al., 2013]{Wang2013d}
Wang, H., Banerjee, A., Hsieh, C.-J., Ravikumar, P.~K., and Dhillon, I.~S.
  (2013).
\newblock {Large Scale Distributed Sparse Precision Estimation}.
\newblock In Burges, C. J.~C., Bottou, L., Welling, M., Ghahramani, Z., and
  Weinberger, K.~Q., editors, {\em Advances in Neural Information Processing
  Systems 26}, pages 584--592. Curran Associates, Inc.

\bibitem[Witten et~al., 2011]{Witten2011}
Witten, D.~M., Friedman, J.~H., and Simon, N. (2011).
\newblock {New insights and faster computations for the graphical lasso}.
\newblock {\em Journal of Computational and Graphical Statistics},
  20(4):892--900.

\bibitem[Yu et~al., 2018]{Yu2018}
Yu, D., Lee, S.~H., Lim, J., Xiao, G., Craddock, R.~C., and Biswal, B.~B.
  (2018).
\newblock {Fused lasso regression for identifying differential correlations in
  brain connectome graphs}.
\newblock {\em Statistical Analysis and Data Mining: The ASA Data Science
  Journal}, 11(5):203--226.

\bibitem[Yuan and Lin, 2007]{Yuan2007}
Yuan, M. and Lin, Y. (2007).
\newblock {Model selection and estimation in the Gaussian graphical model}.
\newblock {\em Biometrika}, 94:19--35.

\end{thebibliography}

\end{document}